\documentclass[10pt,journal,compsoc]{IEEEtran}
\newtheorem{definition}{Definition}
\newtheorem{strategy}{Strategy}    % gws
\newtheorem{theorem}{Theorem}
\newenvironment{proof}{\begin{IEEEproof}}{\end{IEEEproof}}
\usepackage{multirow}
\usepackage{amsmath}

\ifCLASSOPTIONcompsoc
  \usepackage[nocompress]{cite}
\else
  \usepackage{cite}
\fi

\ifCLASSINFOpdf
   \usepackage[pdftex]{graphicx}

\else
\fi

\usepackage{algorithmic}
\usepackage{algorithm}
\usepackage{array}

  \usepackage[caption=false,font=footnotesize,labelfont=sf,textfont=sf]{subfig}
  \usepackage[caption=false,font=footnotesize]{subfig}

\newcommand\MYhyperrefoptions{bookmarks=true,bookmarksnumbered=true,
	pdfpagemode={UseOutlines},plainpages=false,pdfpagelabels=true,
	colorlinks=true,linkcolor={blue},citecolor={blue},urlcolor={blue},
	pdftitle={Targeted High-Utility Sequence Querying},%<!CHANGE!
	pdfsubject={Typesetting},%<!CHANGE!
	pdfauthor={Chunkai Zhang},%<!CHANGE!
	pdfkeywords={Targeted-oriented querying; target sequence; utility mining; pruning strategies.}}%<^!CHANGE!
\ifCLASSINFOpdf
\usepackage[\MYhyperrefoptions,pdftex]{hyperref}
\else
\usepackage[\MYhyperrefoptions,breaklinks=true,dvips]{hyperref}
\usepackage{breakurl}
\fi

\usepackage{url}

\begin{document}
%
% paper title
\title{TUSQ: Targeted High-Utility Sequence Querying}

\author{% Chunkai Zhang, Zilin Du, Quanjian Dai, Wensheng Gan*, Jian Weng, and Philip S. Yu

Chunkai Zhang,~\IEEEmembership{Member,~IEEE,}
 Zilin Du, Quanjian Dai, Wensheng Gan,~\IEEEmembership{Member,~IEEE,}\\
  Jian Weng,~\IEEEmembership{Member,~IEEE,}
    and Philip S. Yu,~\IEEEmembership{Life Fellow,~IEEE,}
	
	\IEEEcompsocitemizethanks{\IEEEcompsocthanksitem Chunkai Zhang, Zilin Du, and Quanjian Dai are with the Department of Computer Science and Technology, Harbin Institute of Technology (Shenzhen), Shenzhen 518055, China.
	
	\IEEEcompsocthanksitem Wensheng Gan is with the College of Cyber Security, Jinan University, Guangzhou 510632, China; and with Guangdong Artificial Intelligence and Digital Economy Laboratory (Guangzhou), China.

	\IEEEcompsocthanksitem Jian Weng is with the College of Cyber Security, Jinan University, Guangzhou 510632, China.

	\IEEEcompsocthanksitem Philip S. Yu is with the Department of Computer Science, University of Illinois at Chicago, IL, USA.}
	
	\thanks{Corresponding author:  Wensheng Gan, E-mail: wsgan001@gmail.com}
	%\thanks{Manuscript received XX, 2021; revised XX.}
}

\IEEEtitleabstractindextext{%
\begin{abstract}
	
Significant efforts have been expended in the research and development of a database management system (DBMS) that has a wide range of applications for managing an enormous collection of multisource, heterogeneous, complex, or growing data. Besides the primary function (i.e., create, delete, and update), a practical and impeccable DBMS can interact with users through information selection, that is, querying with their targets. Previous querying algorithms, such as frequent itemset querying and sequential pattern querying (SPQ) have focused on the measurement of frequency, which does not involve the concept of utility, which is helpful for users to discover more informative patterns. To apply the querying technology for wider applications, we incorporate utility into target-oriented SPQ and formulate the task of targeted utility-oriented sequence querying. To address the proposed problem, we develop a novel algorithm, namely targeted high-utility sequence querying (TUSQ), based on two novel upper bounds \textit{suffix remain utility} and \textit{terminated descendants utility} as well as a vertical Last Instance Table structure. For further efficiency, TUSQ relies on a projection technology utilizing a compact data structure called the targeted chain. An extensive experimental study conducted on several real and synthetic datasets shows that the proposed algorithm outperformed the designed baseline algorithm in terms of runtime, memory consumption, and candidate filtering.

\end{abstract}

\begin{IEEEkeywords}
	Targeted-oriented querying, target sequence, utility mining, pruning strategies.
\end{IEEEkeywords}}

% make the title area
\maketitle

\IEEEdisplaynontitleabstractindextext

\IEEEpeerreviewmaketitle

%%%%%%%%%%%%%%%%%%%%%%%%%%%%%%%%%%%%%%%
\IEEEraisesectionheading{
\section{Introduction} \label{sec:introduction}}

\IEEEPARstart{A}{s} a fundamental research topic in the domain of Knowledge Discovery in Database (KDD), frequent itemset mining \cite{agrawal1994fast} (FIM) aims to extract all frequent itemsets from a transaction database including a set of itemsets with respect to the frequency measure, namely support. Let us consider an example in the market basket analysis task. The itemset $\{\textit{ham}\ \textit{bread}\ \textit{milk}\}$ can be discovered owing to its high frequency, which indicates that a large number of customers buy the three commodities together. To better understand customer behavior, business-oriented rules must be explored based on the frequent itemsets that are obtained with respect to the confidence measure. In the aforementioned case, if a large number of customers who buy ham and bread also buy milk, then the desired rule has the form $\{\textit{ham}\ \textit{bread}\} \Rightarrow \{\textit{milk}\}$. Rules such as this are called association rules \cite{agrawal1994fast}, which can significantly reflect a causal relationship. Till now, two technologies, i.e., FIM and association rule mining (ARM), have been extensively studied \cite{hipp2000algorithms,zaki2001spade}. It is noted that FIM is the first step of ARM; consequently, they are often mentioned together owing to the progressive relationship. The results with occurrence frequency are beneficial for understanding customer behavior and support decisions for copromoting products and offering discounts to increase sales \cite{fournier2017survey}.

As a significant form of big data, sequence data is gradually being considered in several real-world applications, such as sequential expression of genes, web browsing history, and market expenditure records. The elements in a sequence are arranged in chronological order according to their occurrence time. However, FIM algorithms cannot handle sequence data. This challenge motivates a new research issue named sequential pattern mining (SPM) \cite{agrawal1995mining,ayres2002sequential}. By maintaining the order of elements in a sequence, SPM is able to extract all frequent subsequences, which are regarded as sequential patterns, existing in a sequence database. For example, SPM algorithms can be used to identify patterns such as $<$\{\textit{phone}\}, \{\textit{phone shell}\}$>$ in market basket analysis, indicating that these commodities are frequently bought by customers in a chronological order. Obviously, SPM considers the occurrence and sequential relationship information that is ignored by FIM, which intrinsically helps discover more informative patterns.

In several real-life situations, retailers prefer to identify patterns that can yield a high profit rather than those bought frequently. To address this problem, the tasks of high-utility itemset mining (HUIM) \cite{liu2005two,fournier2014fhm} and high-utility SPM (HUSPM) \cite{ahmed2010novel,shie2011mining,gan2021survey} were generalized from FIM and SPM, respectively, with the integration of utility. HUIM is a prominent solution for analyzing quantitative data with a series of itemsets (i.e., transactions) where items are quantified by utility values. Correspondingly, HUSPM can handle databases with a set of quantitative sequences. For brevity, the two technologies are collectively referred to as utility mining. In practice, in the problem of utility mining, each type of item is associated with an integer named external utility, which usually represents the unit profit of the item. In addition, each occurrence of an item is embedded in an integer called internal utility, which usually denotes the sales quantity in a purchase. The main goal of utility mining is to extract patterns with high utility, that is, to discover certain patterns that yield a high profit.

Significant efforts have been expended in the research and development of a database management system (DBMS) \cite{mccarthy1989architecture,stonebraker1991postgres} that has a wide range of applications for managing an enormous collection of multisource, heterogeneous, complex, or growing data. Besides the primary functions (i.e., create, delete, and update), a practical and impeccable DBMS can interact with users by information selection, that is, querying according to their targets. The aforementioned five technologies (i.e., FIM, ARM, SPM, HUIM, and HUSPM) were designed to discover the complete set of patterns that satisfy the minimum utility threshold as defined by the user in a database; however, they do not allow the user to perform targeted queries. To satisfy the requirements of users, a series of frequency-based methods that can search for specific goal items were proposed. Till now, target-oriented frequent itemset querying \cite{shabtay2018guided}, association rule querying \cite{kubat2003itemset,fournier2013meit,abeysinghe2017query}, and sequential pattern querying (SPQ) \cite{chiang2003goal,chueh2010mining,chand2012target} have performed significant roles in querying in the database. The three target-oriented technologies can efficiently excavate patterns and rules involving a subset of certain items, such as targeted queries, and have shown significant potential in several real-life situations \cite{abeysinghe2018query}.

Previous querying algorithms have focused on the measurement frequency, which does not involve the concept of utility, which is helpful for users to discover more informative patterns. To apply the querying technology to wider applications, we incorporated utility into target-oriented SPQ and proposed a novel problem named targeted utility-oriented sequence querying. The task of targeted utility-oriented sequence querying is predominantly to extract subsequences, as targeted queries, that not only contain the predefined target but also have a high utility value with respect to a specified minimum utility threshold. These customized patterns (i.e., targeted queries) with auxiliary knowledge can benefit a number of applications. For example, in a retail store, \textit{milk} is heavily hoarded and must be sold on sale at present. Then, the targeted utility-oriented sequence querying method can be used to discover high-utility patterns pertaining to the item \textit{milk}. Let us assume the pattern $<$\{\textit{milk} \textit{bread}\}, \{\textit{ham}\}$>$ is extracted; then, the decision-makers can adopt a cross-marketing strategy that promotes \textit{milk} and \textit{bread} together. For further profit, they can also discount \textit{ham} over a period of time. Obviously, these extracted target-related high-utility patterns are helpful for retailers for making decisions and performing loss-leader analysis. Moreover, targeted utility-oriented sequence querying can also be applied to the mobile commerce environment, gene expression analysis, time-aware recommendation, etc.  

Several challenges exist in solving the proposed problem of targeted utility-oriented sequence querying, which can be described as follows. First, the method to calculate the utility is moderately different from that used for measuring frequency. When calculating utility values, identifying all occurrences of a candidate in a database is required. This also leads to the famous apriori \cite{agrawal1994fast, agrawal1995mining} property, which is not suitable for the utility mining problem, indicating that the existing pruning strategies utilizing frequency anti-monotonicity in frequency-based approaches cannot be directly applied to our problem. Second, the inherent sequential order of elements in sequences results in a critical combinatorial explosion of the search space. Thus, it is significantly difficult to mine targeted queries without designing tight upper bounds of utilities for reducing candidates that are not a desired result of the query, especially in large-scale databases. Third, in general, the method identifies the desired queries by enumerating candidates in alphabetical order. However, certain prefixes cannot be extended to queries containing the target. Identifying a method to terminate enumerating based on these prefixes as early as possible is a challenging problem. Fourth, in general, a common way for saving the memory consumption is to generate projected databases that are small-scale and useful in reducing the scope of the scan. Thus, designing a compact data structure to save projected databases is also a challenging issue.

With the rapid development of networking, data storage, and data collection capacity, big data are now rapidly expanding in all science and engineering domains, including physical, biological, and biomedical sciences. The volume of data handled by DBMSs are becoming larger and larger. Thus, there is an urgent requirement for algorithms that can process massive databases efficiently and provide acceptable querying performance. The primary contributions of this paper can be summarized as follows: 

\begin{itemize}
	
	\item To the best of our knowledge, this is the first paper that incorporates the concept of utility into target-oriented SPQ. This work formulates the problem of targeted utility-oriented sequence querying aiming to discover the complete set of desired queries (i.e., utility-driven targeted queries) that are more informative and useful for users.
	
	\item We successfully solved the above querying problem using a novel algorithm named targeted high-utility sequence querying (TUSQ) with a preprocessing strategy. A compact data structure, called targeted utility, is applied to store the necessary information in the projected databases.  
	
	\item For further efficiency improvement, we developed two novel upper bounds, i.e., suffix remain utility (\textit{SRU}) and terminated descendants utility (\textit{TDU}), as well as a Last Instance Table (LI-Table) structure for facilitating the calculation of upper bounds. Based on the anti-monotonous properties of upper bounds, two pruning strategies were designed to accelerate the mining process. 
	
	\item Both real-life and synthetic datasets were used in the evaluation. By conducting substantial experiments, the TUSQ algorithm shows its effectiveness and high efficiency for mining all desired queries with a user-defined minimum utility threshold from databases.
	
\end{itemize}

The remainder of this paper is organized as follows. In Section \ref{sec:relatedwork}, certain related works are briefly reviewed. Section \ref{sec:preliminaries} provides cetain essential definitions and formulates the problem to be solved. The proposed method, TUSQ, with several effective upper bounds and compact data structures, is described in Section \ref{sec:method}. A series of experiments were conducted and the results are presented to illustrate the performance of the proposed algorithms in Section \ref{sec:experiments}. Finally, the conclusions are drawn in Section \ref{sec:conclusion}.

\section{Related work}  \label{sec:relatedwork}

% In this section, certain related studies on conventional pattern mining and target-oriented querying are briefly reviewed.

\subsection{Conventional pattern mining}

Pattern mining \cite{pei2004mining,fournier2017survey} has emerged as a well-studied data mining technology, which is applied in several real-world situations, such as website clickstream analysis \cite{li2008fast}, gene regulation \cite{zihayat2017mining}, image classification \cite{fernando2012effective}, and network traffic analysis \cite{brauckhoff2009anomaly}. The problem pertaining to FIM and ARM was first discussed by Agrawal et al. \cite{agrawal1994fast}. To address this issue, they designed the famous Apriori algorithm utilizing the anti-monotonicity of frequency. For further efficiency, a series of algorithms \cite{han2004mining, zaki2000scalable} that can avoid exploring the complete search space of all possible itemsets were developed by adopting pruning strategies. Itemset mining is a significantly active research field, where certain extension problems of FIM, such as incremental mining \cite{leung2007cantree}, stream mining \cite{chang2003finding}, and fuzzy mining \cite{hong2004fuzzy}, were developed and widely studied. Then, as a prominent solution for analyzing sequential data, the task of SPM considering the sequential ordering of itemsets was proposed. Inspired by the Apriori principle, a set of Apriori-based algorithms such as AprioriAll \cite{agrawal1995mining} and GSP \cite{srikant1996mining} were designed for handling the issue. Till date, there have been several existing studies \cite{pei2004mining,zaki2001spade} in the literature about determining sequential patterns from sequence databases when the minimum support threshold is provided. More details about FIM, ARM, and SPM can be obtained in \cite{fournier2017surveyi,fournier2017surveys}.

Utility is a measure of the satisfaction or interests that a consumer obtains from buying a good or service, which can help retailers maximize their benefits. Recently, utility mining (e.g., HUIM, HUSPM) has emerged as a popular issue owing to the thorough consideration of quantities, profits, and chronological order of items. Chan et al. \cite{chan2003mining} first incorporated the concept of utility into the FIM and proposed the task of HUIM. To achieve better performance for mining high-utility itemsets, several efficient algorithms such as UP-Growth \cite{qu2019efficient} and HUI-Miner \cite{tseng2012efficient}, which adopt novel pruning strategies and compact data structures, were designed. In addition to transaction data, utility mining also includes sequential data mining. A large number of previous works on HUSPM \cite{ahmed2010novel, wang2016efficiently, gan2020proum} have focused on designing efficient data structures and algorithms that achieve better performance. Among them, the well-known HUS-Span \cite{wang2016efficiently} explores the search space of patterns in the alphabetical order. It also adopts two tight upper bounds, i.e., prefix extension utility and reduced sequence utility, and two corresponding pruning strategies. Moreover, it utilizes the utility chain to store necessary information for calculation. Because of its high efficiency, we designed the basic algorithm as the baseline based on HUS-Span. Currently, the field of HUSPM is still actively researched, and several efficient methods \cite{gan2020fast}, extension problems \cite{gan2021utility,gan2019utility}, new real-life applications \cite{wang2018incremental}, and privacy preserving utility mining \cite{gan2018privacy} have been determined. A comprehensive survey of the present development of utility mining was reported by Gan et al. \cite{gan2021survey}.

\subsection{Target-oriented querying}

All of the aforementioned methods focus on extracting the complete set of patterns that satisfy a predefined threshold. To filter out unwanted information, target-oriented querying algorithms propose a different solution to the problem of pattern mining. The idea is that instead of mining a large number of patterns that may not all be useful, users could input a single target at a time and discover the patterns containing the target as targeted queries. In previous studies, several target-oriented querying approaches were proposed according to the measurement frequency. These interactive methods can return queries with the user-defined targets.

Kubat et al. \cite{kubat2003itemset} are one of the earliest researchers interested in applications where the users wish to present specialized queries in a transaction database. With the designed Itemset-Tree, it can extract all the rules (i.e., targeted queries) that have a user-specified itemset as antecedent according to a minimum confidence and support measurement. This task belongs to the domain of target-oriented querying. Fournier et al. \cite{fournier2013meit} improved the Itemset-Tree, which was optimized for quickly answering queries about itemsets during the operation process in different tasks. Their designed structures can be updated incrementally with new transactions.  Shabty et al. \cite{shabtay2018guided} designed a novel method called Guided FP-growth (GFP-growth) for multitude-targeted mining. The fast and generic tool GFP-growth can determine the frequency of a given large list of itemsets, which serve as the targets, in a large dataset from an FP-tree \cite{han2000mining} based on Target Itemset Tree. Recently,  a query-constraint-based ARM model  \cite{abeysinghe2017query,abeysinghe2018query} was developed for exploratory analysis of diverse clinical datasets integrated in the National Sleep Research Resource. It is important to consider the sequential ordering of itemsets in real-life applications. To handle the sequence data, Chueh et al. \cite{chueh2010mining} defined the target-oriented sequential pattern as a sequential pattern with a concerned itemset (i.e., target) at the end of the pattern. They presented an algorithm reversing the original sequences to discover target-oriented sequential patterns with time intervals. Utilizing the definition of the target-oriented sequential pattern, Chand et al. \cite{chand2012target} proposed a novel SPM approach, which not only determined the existence of a pattern but also checked whether the pattern was target oriented and also satisfied the recency and monetary constraints. Moreover, a new SPM algorithm called goal-oriented algorithm \cite{chiang2003goal} was proposed to discover the transaction activity before losing the customer. It can handle the problem of judging whether a customer is leaving by mining sequential patterns toward a specific goal.

The aforementioned methods can output targeted queries according to the frequency. However, no research has incorporated concept utility into target-oriented querying. The existing HUSPM algorithms obtain exhaustive lists of HUSPs, some of which may be useless. This motivated us to develop an efficient target-oriented querying method for mining HUSPs containing a target sequence from quantitative sequence databases.

\section{Preliminaries}   \label{sec:preliminaries}

In this section, we first present the essential concepts and definitions in our proposal. Based on these preliminaries, we formulated a novel problem of targeted utility-oriented sequence querying that must be addressed.

\subsection{Definitions}

\begin{definition}[itemset and sequence]
	Let the finite set $I$ = \{$i_{1}$, $i_{2}$, $\cdots$, $i_{N}$\} be a set of items that may appear in the database. We say that $x$ is an itemset if $x$ is a nonempty subset of $I$, i.e., $x \subseteq I$. Assuming that $x$ contains $m$ items, the length of $x$ can be denoted as $|x|$ = $m$. $s$ = $<$$x_{1}$, $x_{2}$, $\cdots$, $x_{n}$$>$ is called a sequence consisting of a series of itemsets that satisfies the condition $x_{k}\subseteq I$ for $1 \leq k \leq n$, which are arranged in order. We define the length of $s$, denoted as $|s|$, as the sum of the length of each element (i.e., itemset); thus, $|s|$ = $\sum_{k = 1}^{n}|x_{k}|$. $S$ with a length of $l$ is called an $l$-sequence. 
\end{definition}

For example, given a set $I$ = \{$a,b,c,d,e,f$\}, \{$a$\}, and \{$af$\} are itemsets. Further, $s$ = $<$$\{a\}$, $\{af\}$, and $\{de\}$$>$ is a sequence with three itemsets, whose lengths are 1, 2, and 2, respectively. Thus, the length of $s$ is $|s|$ = 1 + 2 + 2 = 5.

\begin{definition}[subsequence]	
	Given two sequences $s$ = $<$$x_{1}$, $x_{2}$, $\cdots$, $x_{m}$$>$ and $s'$ = $<$$x'_{1}$, $x'_{2}$, $\cdots$, $x'_{n}$$>$, if there exists $m$ integers $1 \leq k_{1}$ $< k_{2}$ $< \cdots < k_{m} \leq n$, such that $x_{1}$ $\subseteq$ $x'_{k_{1}}$, $x_{2}$  $\subseteq$ $x'_{k_{2}}$, $\cdots$, $x_{m}\subseteq x'_{k_{m}}$, then $s$ is said to be a subsequence of $s'$, which is denoted as $s \subseteq s'$.
\end{definition}

For example, consider three sequences $s_1$ = $<$$\{a\}$, $\{af\}$, $\{de\}$$>$, $s_2$ = $<$$\{cde\}$, $\{df\}$$>$, and $s'$ = $<$$\{bc\}$, $\{a\}$, $\{acf\}$, $\{def\}$$>$. We say $s_1$ is a subsequence of $s'$, while $s_2$ is not a subsequence of $s'$ as the first itemset of $s_2$ (i.e., $\{cde\}$) is not a subset of any itemset of $s'$.

\begin{definition}[$q$-item, $q$-itemset and $q$-sequence]
	A quantitative sequence ($q$-sequence) database $D$ is formally defined as follows. A quantitative item ($q$-item) is an item $i \in I$ associated with a quantity $q \in \mathcal{R}^+$, which can be represented as a pair ($i$:$q$). Moreover, each item $i \in I$ has its own unit profit. Generally, we refer to the quantity and unit profit as the internal and external utilities, respectively. A quantitative itemset ($q$-itemset) $X$ = \{($i_{1}$:$q_{1}$) ($i_{2}$:$q_{2}$) $\cdots$ ($i_{m}$:$q_{m})$\} is a set of $q$-items. Without loss of generality, items/$q$-items in an itemset/$q$-itemset are arranged in the alphabetical order in the remainder of this paper. A quantitative sequence ($q$-sequence) $S$ = $<$$ X_{1}$, $X_{2}$, $\cdots$, $X_{n}$$>$ is a list of $q$-itemsets, where each $q$-itemset is arranged in ascending order according to the occurrence time. A series of $q$-sequences with their identifiers are composed of a $q$-sequence database $D$ = \{$S_{1}$, $S_{2}$, $\cdots$, $S_{M}$\}.
\end{definition}

To illustrate the important definitions, a quantitative sequence database with five $q$-sequences and six distinct items, as listed in Table \ref{table1}, which is used as the running example. The external utility of each item is listed in Table \ref{table2}. For instance, ($a$:$1$) is a $q$-item, which represents that item $a$ has a quantitative value of $1$. \{($a$:$1$) ($b$:$3$)\} is a $q$-itemset with two $q$-items ($a$:$1$) and ($b$:$3$), and $S_{1}$ = $<$\{($a$:$1$) ($b$:$3$)\}, \{($c$:$1$) ($e$:$2$)\}, \{($c$:$4$) ($d$:$1$)\}$>$ in Table \ref{table1} is a $q$-sequence containing three $q$-itemsets \{($a$:$1$) ($b$:$3$)\}, \{($c$:$1$) ($e$:$2$)\}, and \{($c$:$4$) ($d$:$1$)\}.

\begin{table}[!t]
	\centering
	\caption{Running example of a $q$-sequence database}
	\label{table1}
	\begin{tabular}{|c|c|}  
		\hline 
		\textbf{SID} & \textbf{$q$-sequence} \\
		\hline  
		\(\textit{S}_{1}\) & $<$\{(\textit{a}:1) (\textit{b}:3)\}, \{(\textit{c}:1) (\textit{e}:2)\}, \{(\textit{c}:4) (\textit{d}:1)\}$>$ \\ 
		\hline
		\(\textit{S}_{2}\) & 
		$<$\{(\textit{a}:3)\}, \{(\textit{c}:1) (\textit{d}:1) (\textit{e}:2)\}, 
		\{(\textit{a}:1)\}, 
		\{(\textit{c}:3) (\textit{e}:1)\}$>$ \\
		\hline 
		\(\textit{S}_{3}\) & $<$\{(\textit{b}:2) (\textit{c}:1)\}, \{(\textit{f}:2)\}, 
		\{(\textit{a}:3) (\textit{d}:2)\}$>$ \\  
		\hline  
		\(\textit{S}_{4}\) & 
		$<$\{(\textit{a}:1)\}, 
		\{(\textit{c}:3) (\textit{e}:2)\},
		\{(\textit{c}:2) (\textit{e}:1)\},
		\{(\textit{c}:2)\}$>$ \\
		\hline
		\(\textit{S}_{5}\) & 
		$<$\{(\textit{a}:2)\},
		\{(\textit{c}:1) (\textit{e}:1)\}, \{(\textit{c}:2) (\textit{e}:1)\},
		\{(\textit{f}:2)\}$>$ \\
		\hline
	\end{tabular}
\end{table}

\begin{table}[!t]
	\caption{External utility table}
	\label{table2}
	\centering
	\begin{tabular}{|c|c|c|c|c|c|c|}
		\hline
		\textbf{Item}	    & \textit{a}	& \textit{b}	& \textit{c}	& \textit{d}	& \textit{e}	& \textit{f} \\ \hline 
		\textbf{External utility}	& \$2 & \$1& \$1 & \$3 & \$1 & \$2 \\ \hline
	\end{tabular}
\end{table}

\begin{definition}[match]
	Let us consider an itemset $x$ = \{$ i'_{1}$, $i'_{2}$, $\cdots$, $i'_{m}$\} and a $q$-itemset $X$ = \{($i_{1}$:$q_{1}$) ($i_{2}$:$q_{2}$)$\cdots$($i_{m}$:$q_{m})$\}. We say $x$ matches $X$, denoted as $x \sim X$, if and only if $i'_{k}$ = $i_{k}$ for $1 \leq k \leq m$. On this basis, given a sequence $s$ = $<$$x_{1}$, $x_{2}$, $\cdots$, $x_{n}$$>$, and a $q$-sequence $S$ = $<$$ X_{1}$, $X_{2}$, $\cdots$, $X_{n}$$>$, we say $s$ matches $S$, denoted as $s \sim S$, if and only if $x_{k}\sim X_{k}$ for $1 \leq k \leq n$. 
\end{definition}

For example, \{$ab$\} matches the first itemset of $S_{1}$ and sequence $<$$\{ab\},\{ce\},\{cd\}$$>$ matches $S_{1}$.

\begin{definition}[instance]
	Given a sequence $s$ = $<$$x_{1}$, $x_{2}$, $\cdots$, $x_{m}$$>$ and a $q$-sequence $S$ = $<$$X_{1}$, $X_{2}$, $\cdots$, $X_{n}$$>$, where $m \leq n$, $s$ has an instance in $S$ at position $p$: $<$$k_{1}$, $k_{2}$, $\cdots$, $k_{m}$$>$ if and only if there exists $m$ integers $1 \leq k_{1}$ $< k_{2} $ $< \cdots$ $< k_{m} \leq n$ such that $x'_{v}\sim X_{k_{v}} \land x_{v} \subseteq x'_{v} $ for $1 \leq v \leq m$. Then, we denote $P(s,S)$, where each element is an integer sequence, as the set of instance positions of $s$ in $S$. 
\end{definition}

For example, $<$$\{a\},\{c\}$$>$ has three instances in $S_{2}$ at positions $p_{1}$: $<$1, 2$>$, $p_{2}$: $<$1, 4$>$ and $p_{3}$: $<$3, 4$>$. Thus, $P($$<$$\{a\}$, $\{c\}$$>$, $S_{2}$$)$ = \{$<$1, 2$>$, $<$1, 4$>$, $<$3, 4$>$\}.

\begin{definition}[contain]
	Let us consider a sequence $s$. Here, $s$ is said to be contained in $S$, if and only if $s$ has at least one instance in $S$, which can be denoted as $s' \sim S$ $\land s$ $\subseteq s'$. In the subsequent portions of this paper, $s \sqsubseteq S$ is used to indicate that $s' \sim S$ $\land s$ $\subseteq s'$ for convenience.
\end{definition}

For example, sequence $s$ = $<$$\{a\}$, $\{cd\}$$>$ has an instance in $S_{2}$; therefore, we say that $S_{2}$ contains $s$ and denote it as $s \sqsubseteq S_{2}$. Subsequently, we define the related calculation methods of utility values in different situations. As it can be observed in the $q$-sequence database listed in Table \ref{table1}, each occurrence of an item, $i$, is embedded in an internal utility, denoted as $q(i, j, S)$, which is the quantitative measure for ($i$:$q$) within the $j$-th $q$-itemset of a $q$-sequence $S$. In addition, each type of item is associated with an external utility, denoted as $p(i)$, which is listed in Table \ref{table2}.

\begin{definition}[utility of $q$-item, $q$-itemset and $q$-sequence]
	Given a $q$-sequence $S$, let $u(i, j, S)$ denote the utility of a $q$-item $i$ within the $j$-th $q$-itemset in $S$, which is specified as $u(i, j, S)$ = $q(i, j, S)$ $\times$ $p(i)$. Moreover, the utility of a $q$-itemset/$q$-sequence is defined as the sum of utility values of elements (i.e., $q$-items/$q$-itemsets) contained. 
\end{definition}

For example, the utility of the $q$-item ($a$:1) within the first $q$-itemset of $S_{1}$ can be calculated as $u(a, 1, S_{1})$ = $q(a, 1, S_{1})$ $\times$ $p(a)$ = 1 $\times$ \$2 = \$2. Moreover, the utility of the first $q$-itemset of $S_{1}$ is $u(1,S_{1})$ = \$2 + \$3 = \$5; the utility of $S_{1}$ is $u(S_{1})$ = \$5 + \$3 + \$7 = \$15.

\begin{definition}[utility of instance]
	Given an itemset $x$, the sequence $s$ = $<$$x_{1}$, $x_{2}$, $\cdots$, $x_{m}$$>$ and $q$-sequence $S$ in a $q$-sequence database $D$, we formalize the calculation of itemset utility in the $j$-th $q$-itemset in $s$ as $u(x, j, S)$ = $\sum_{\forall i \in x}^{}{u(i, j, S)}$. It is assumed that $s$ has an instance in $S$ at position $p$: $<$$k_{1}$, $k_{2}$, $\cdots$, $k_{m}$$>$. The utility of this instance can be specified as $u(s, p, S)$ = $\sum_{j=1}^{m}u(x_j, k_{j}, S)$, where $x_j$ is the $j$-th itemset of $s$. 
\end{definition}

For example, $u$($<$\{$a$\}, \{$c$\}$>$, $<$1, 2$>$, $S_{1}$) = \$2 + \$1 = \$3; $u$($<$\{$a$\}, \{$c$\}$>$, $<$1, 3$>$, $S_{1}$) = \$2 + \$4 = \$6. Obviously, a sequence has multiple utility values in a $q$-sequence, which is noticeably different from frequent-based SPM. We denote the utility of $s$ in $S$ as $u(s,S)$. It can be calculated as $u(s,S)$ = $\max\{u(s, p, S)| \forall p \in P(s,S)\}$. For example, $u$($<$\{$a$\}, \{$c$\}$>$, $S_{1}$) = max\{\$3, \$6\} = \$6. The utility of $s$ in $D$ is the overall utility value in all $q$-sequences, which is specified as $u(s)$ = $\sum_{\forall S \in D}^{}u(s,S)$. Let us consider the $q$-sequence database listed in Table \ref{table1}; we obtain $u$($<$\{$a$\}, \{$c$\}$>$) = \$6 + \$9 + \$5 + \$6 = \$26.

\subsection{Problem statement}
\label{Problem statement}

Before we formulate the problem of targeted utility-oriented sequence querying, the $q$-sequence database $D$ can be divided into two sets. All $q$-sequences containing the target sequence $T$ are arranged in the first set, denoted as $D_T$ = \{$S\mid S \in D \land T \sqsubseteq S$\}, or are placed in the other set. 

\begin{definition}[utility-driven targeted query]
	\label{utility-driven targeted query}
	In a $q$-sequence database $D$, given a target sequence $T$ and a minimum utility threshold $\xi$, a sequence $s$ is called a utility-driven targeted query, if and only if $T$ is a subsequence of $s$ (i.e., $T \subseteq s$), and $u(s) \ge \xi $ $\times$ $u(D_T)$.
\end{definition}

\textbf{Problem Statement:} Given a $q$-sequence database with a utility table, the problem of targeted utility-oriented sequence querying is to discover the complete set of utility-driven targeted queries (UTQs) when users present a minimum utility threshold and a target sequence.

To elaborately illustrate the problem of targeted utility-oriented sequence querying, an example is given as follows. In the following sections of this paper, we set $T$ = $<$\{$a$\}, \{$c e$\}, \{$c$\}$>$ as the target sequence. Among the five $q$-sequences listed in Table \ref{table1}, $S_{3}$ does not contain $T$, whereas the other four contain it. Thus, $S_{3}$ can be filtered out, and we obtain $D_T$ = \{$S_{1}$, $S_{2}$, $S_{4}$, $S_{5}$\}. It can be calculated that $u(D_{T})$ = \$58. When the minimum utility threshold $\xi$ is set to 30\%, the discovered UTQs are $<$\{$a$\}, \{$c e$\}, \{$c$\}$>$ and $<$\{$a$\}, \{$c e$\}, \{$c e$\}$>$, with a utility of \$38 and \$19, respectively.

\section{Proposed method}   \label{sec:method}

To the best of our knowledge, no study has focused on the problem of targeted utility-oriented sequence querying. This motivated us to develop the TUSQ method, which can extract the complete set of targeted queries from a $q$-sequence database. This section presents the details of the developed method, which adopts two novel upper bounds and two corresponding pruning strategies to overcome the problem of combinatorial explosion of a search space. To improve the efficiency, we propose a database preprocessing (DPP) strategy and a projection mechanism to reduce the calculation amount.

\subsection{Database preprocessing strategy}

In the task of targeted utility-oriented sequence querying, given a target sequence $T$ and $q$-sequence database $D$, we call the $q$-sequence that does not contain $T$ in $D$ as a \textit{redundant $q$-sequence}. Intuitively, the redundant sequences do not contribute to the utility of the targeted queries that contain $T$; consequently, they can be filtered out before the mining process. Thus, we introduce the following database DPP strategy based on Section \ref{Problem statement}.

\begin{strategy}
	(\textit{DPP strategy}) Given a target sequence $T$ and $q$-sequence database $D$, we can remove the $q$-sequence $S \notin D_T$ from $D$ before the mining process. Then, the TUSQ algorithm is performed on the filtered $q$-sequence database, i.e., $D_{T}$.
\end{strategy} 

\subsection{Pruning strategies}

To facilitate the following discussions, we first present the following definitions.

\begin{definition}[extension]
	\rm Let us consider an $l$-sequence $s$; the extension operations (i.e., $I$-Extension and $S$-Extension) add an item to the end of $s$ and generate an extension sequence $s'$ that is an ($l$+1)-sequence. Appending item $i$ to the last itemset of $s$ is referred to as $I$-Extension, where $s'$ can be denoted as $<$$ s \bigoplus i $$>$. Further, we define the $S$-Extension of sequence $s$ as the addition of $i$ to a new itemset and appending the itemset at the end of $s$, where $s'$ is represented as $<$$s \bigotimes i$$>$.
\end{definition}

\begin{definition}[extension position]\label{extension_position}
	\rm Given a sequence $s$= $<$$x_{1}$, $x_{2}$, $\cdots$, $x_{m}$$>$ and $q$-sequence $S$, it is assumed that $s$ has an instance at position $<$$k_{1}$, $k_{2}$, $\cdots$, $k_{m}$$>$ in $S$. We define the last item of $x_{m}$ as the extension item of $s$ in $S$, and $k_{m}$ is called the extension position of $s$ in $S$.
\end{definition}

For example, $<$\{$a b$\}$>$ can be obtained by performing the $I$-Extension of $<$\{$a$\}$>$, while $<$\{$a$\}, \{$c$\}$>$ is an $S$-Extension sequence of $<$\{$a$\}$>$. For example, $<$\{$a$\}, \{$c$\}$>$ has an instance at $p$: $<$1, 2$>$ in $S_{1}$; therefore, $c$ is an extension of $<$\{$a$\}, \{$c$\}$>$ in $S_{1}$, and the extension position is 2.

\begin{definition}[rest sequence]
	Based on Definition \ref{extension_position}, the rest $q$-sequence of $s$ with respect to the extension position $k_{m}$ in $S$, denoted as $S/_{(s,k_{m})}$, which is defined as a part of $S$, from the item after the extension item of such an instance to the end of $S$. The rest utility is the sum of $q$-items in $S/_{(s,k_{m})}$, which can be specified as $\textit{ru}(s,\ k_{m},\ S)$ = $\sum_{\forall i \in S/_{(s,k_{m})}}^{}u(i)$.
\end{definition}

For example, the rest sequence of the instance of $<$\{$a$\}, \{$c$\}$>$ at $p$: $<$1, 2$>$ in $S_{1}$ is $S/_{(<\{a\}, \{c\}>, 2)}$ = $<$\{($e$:2)\}, \{($c$:4) ($d$:1)\}$>$, and the corresponding rest utility is \textit{ru}($<$\{$a$\}, \{$c$\}$>$, 2, $S_{1}$) = \$2 + \$7 = \$9. Let us assume that the sequence $s$ has a series of instances in a $q$-sequence $S$; then, the set of extension positions of such instances can be represented as \{$ p_{1}$, $p_{2}$, $\ldots$, $p_{n}$\}, where \(p_{1}\) is called the pivot of $s$ in $S$. For example, $<$\{$a$\}, \{$c$\}$>$ has an instance at $p_{1}$: $<$1,2$>$ and $p_{2}$: $<$1, 3$>$ in $S_{1}$, whose extension positions are 2 and 3, respectively; therefore, the pivot of $<$\{$a$\}, \{$c$\}$>$ in $S_{1}$ is 2.

\begin{definition}[prefix and suffix]
	Given two sequences $s$ = $<$$x_{1}$, $x_{2}$, $\ldots$, $x_{m}$$>$, and $r$ = $<$$y_{1}$, $y_{2}$, $\ldots$, $y_{n}$$>$ ($m<n$), $s$ is a prefix of $r$ if it satisfies the following conditions: (1) $x_{i}$ = $y_{i}$ for $1 \leq$ $i \leq m-1$; (2) $x_{m}$ $\subseteq$ $y_{m}$; (3) all the items in $(y_{m}-x_{m})$ are alphabetically bigger than those in $x_{m}$. Moreover, the remaining parts after the prefix $s$ in $r$, i.e., the subsequence $<$$(y_{m}-x_{m})$, $y_{m+1}$, $\ldots$, $y_{n}$$>$, are called the suffix of $r$ with respect to prefix $s$. 
\end{definition} 

For example, $<$\{$a$\}, \{$c$\}$>$ is a prefix of $T$ = $<$\{$a$\}, \{$c e$\}, \{$c$\}$>$; the suffix of $T$ with respect to $<$\{$a$\}, \{$c$\}$>$ is $<$\{$e$\}, \{$c$\}$>$.

\begin{definition}[longest prefix contained by a sequence]
	Given a sequence $s$ = $<$$x_{1}$, $x_{2}$, $\ldots$, $x_{m}$$>$, let us consider a target sequence $T$ = $<$$y_{1}$, $y_{2}$, $\ldots$, $y_{n}$$>$ that has a prefix $\alpha$ = $<$$y'_{1}$, $y'_{2}$, $\ldots$, $y'_{k-1}$, $y'_{k}$$>$, where $1\leq k \leq n$. $\alpha$ is said to be the longest prefix of $T$ contained by $s$, denoted as $\textit{Pre}(T,s)$, if it satisfies the following conditions: (1) there exists $k$ integers $1\leq j_{1} $ $\leq j_{2}$ $\leq \ldots $ $\leq j_{k-1}$ $\leq j_{k} \leq m$, such that $y'_{i} \subseteq x_{j_{i}}$ for $1 \leq i \leq k$, i.e., $\alpha$ is a subsequence of $s$; (2) $(y_{k}-y'_{k})$ is either an empty set, or all items in $(y_{k}-y'_{k})$ are alphabetically bigger than those in $(x_{j_{k}}-y'_{k})$ and $j_{k}$ = $m$; (3) there does not exist a longer prefix $\beta$ of $T$ such that $\beta$ satisfies conditions (1) and (2) simultaneously. Additionally, the remaining part after $\textit{Pre}(T,s)$ in $T$ is defined as the suffix of $T$ with respect to $\textit{Pre}(T,s)$, which can be denoted as $\textit{Suf}(T,s)$.
\end{definition} 

For example, let us consider three sequences $s_{1}$ = $<$\{$a$\}, \{$c$\}$>$, $s_{2}$ = $<$\{$a$\}, \{$c f$\}$>$, $s_{3}$ = $<$\{$a$\}, \{$c$\}, \{$d$\}$>$. Consider the target sequence $T$ = $<$\{$a$\}, \{$c e$\}, \{$c$\}$>$; then, we have \textit{Pre}($T,s_{1})$ = $<$\{$a$\}, \{$c$\}$>$, \textit{Suf}($T,s_{1})$ = $<$\{$e$\}, \{$c$\}$>$; \textit{Pre}($T,s_{2})$ = $<$\{$a$\}$>$, \textit{Suf}($T,s_{2})$ = $<$\{$ce$\}, \{$c$\}$>$; and \textit{Pre}($T,s_{3})$ = $<$\{$a$\}$>$, $\textit{Suf}(T,s_{3})$ = $<$\{$ce$\}, \{$c$\}$>$.

\begin{definition}[promising extension position]
	\rm Given a sequence $s$ = $<$$x_{1}$, $x_{2}$, $\cdots$, $x_{m}$$>$, target sequence $T$, and $q$-sequence $S$, it is assumed that the set of extension positions of all instances of $s$ in $S$ is \textit{EP}($s,S)$ = \{\textit{ep}$_1$, \textit{ep}$_2$, $\cdots$, \textit{ep}$_j$\}. For $\forall ep \in$ $\textit{EP}(s,S)$, if the rest sequence of $S$ with respect to $ep$ contains $\textit{Suf}(T,s)$ (i.e., \textit{Suf}($T,s) \sqsubseteq S/_{(s,ep)}$), then we say $ep$ is a promising extension position. The complete set of promising extension positions of all instances of $s$ in $S$ is denoted as $\textit{PEP}(s,S)$.
\end{definition}

For example, given a sequence $s$ = $<$\{$a$\}, \{$c$\}$>$, we have a set of extension positions \textit{EP}($s,S_{1})$ = $\{2, 3\}$; then, their corresponding rest sequences $S_{1}/_{(s,2)}$ = $<$\{($e$:2)\}, \{($c$:4) ($d$:1)\}$>$ and $S_{1}/_{(s,3)}$ = $<$\{($d$:1)\}$>$ can be obtained. Let us consider the target sequence $T$ = $<$\{$a$\}, \{$c e$\}, \{$c$\}$>$. If we have \textit{Suf}($T,s)$ = $<$\{$e$\}, \{$c$\}$>$, we say that 2 is a promising extension position and \textit{PEP}($s,S_{1})$ = \{2\}, as $S_{1}/_{(s,2)}$ contains $\textit{Suf}(T,s)$, while $S_{1}/_{(s,3)}$ does not.

In the task of conventional HUSPM, the sequential order of itemsets in sequences is considered, which causes the problem of the combinatorial explosion of the search space. In the complete search space, the number of candidates to be examined can be up to $2^{m \times n}$, where $m$ is the number of distinct items and $n$ is the number of itemsets of the longest sequence in the $q$-sequence database \cite{gan2020proum}. The search space is always represented as a data structure named lexicographic quantitative sequence tree (\textit{LQS}-tree), where each node denotes a candidate sequence except the root. More details about the \textit{LQS}-tree can be obtained in Ref. \cite{gan2020proum,zhang2020tkus}. Furthermore, the calculation method of utility is different from that of frequency. This leads to the problem that the apriori property is not satisfied for the utility of sequential patterns, which indicates that we should calculate the utility of all the $2^{m \times n}$ candidates to determine whether they are HUSPs. To address this issue, several upper bounds \cite{yin2012uspan,wang2016efficiently} that overestimate the utility value are proposed to prune the search space. Although these upper bounds and corresponding pruning strategies can significantly reduce the search space, they are not sufficiently efficient in targeted utility-oriented sequence querying. To increase the speed of the mining process, we incorporate the prefix and suffix information of the sequence into the design of two novel upper bounds \textit{SRU} and \textit{TDU}, and introduce the corresponding pruning strategies.

\begin{definition}[suffix remain utility]
	Let us consider a sequence $s$, target sequence $T$, and $q$-sequence $S$. It is assumed that $s$ has an instance in $S$ with an extension position $p$. Note that $\textit{Pre}(T,s)$ is the longest prefix of $T$ contained by $s$, and $\textit{Suf}(T,s)$ is the suffix of $T$ with respect to $\textit{Pre}(T,s)$. The \textit{SRU} of $s$ in $S$ at position $p$ with respect to $T$, denoted as \textit{SRU}( $s,T,p,S)$, is defined as  \textit{SRU}( $s, T, p, S) $ = $ u(s, p, S) $ + \textit{ru}( $s, p, S) $ if and only if \textit{ru}$(s, p, S) > 0$ $\land$ \textit{Suf}( $T,s) \sqsubseteq S/_{(s,p)}$ can be established; otherwise, \textit{SRU}$(s, T, p, S)$ = 0.
\end{definition} 

In the above definition, the constraint condition \textit{Suf}$(T,s)$ $\sqsubseteq$ $S/_{(s,p)}$ indicates that $\textit{SRU}(s,T,p,S) > 0$ only if $p$ is a promising extension position. Note that $EP$ is the set of all extension positions of $s$ in $S$. The \textit{SRU} of $s$ in $S$ with respect to $T$, denoted as $\textit{SRU}(s,T,S)$, is defined as
\[\textit{SRU}(s,T,S) = \mathop{max}\limits_{\forall p \in \textit{EP}}{\textit{SRU}(s,T,p,S)}.\]

The \textit{SRU} of $s$ with respect to $T$ in a $q$-sequence database $D$, denoted as $\textit{SRU}(s,T,D)$, is defined as
\[\textit{SRU}(s,T) = \sum_{\forall S \in D \land t \sqsubseteq S}^{}\textit{SRU}(s,T,S).\]

For instance, consider the sequence $s$ = $<$\{$a$\}, \{$ce$\}$>$, which has three instances at position $<$1, 2$>$, $<$1, 4$>$, and $<$3, 4$>$ in $S_{2}$. Then, we obtain \textit{SRU}$(s,T$, $<$1,2$>$, $S_{2})$ = \$15; both \textit{SRU}$(s,T$, $<$1, 4$>$, $S_{2})$ and \textit{SRU}$(s,T$, $<$3,4$>$, $S_{2})$ are equal to \$0 as \textit{ru}($s$, 4, $S_{2}$) = \$0. Moreover, $s$ has two instances at positions $<$1, 2$>$ and $<$1, 3$>$ in $S_{5}$. We have \textit{SRU}$(s,T$, $<$1, 2$>$, $S_{5})$ = \$13; further, \textit{SRU}$(s,T$, $<$1, 3$>$, $S_{5})$ = \$0 because 3 is not a promising extension position. Consider one more example where $s$ has two instances at positions $<$1, 2$>$ and $<$1, 3$>$ in $S_{4}$. We have \textit{SRU}$(s,T,$ $<$1, 2$>$, $S_{4})$ = \$12 and \textit{SRU}$(s,T$, $<$1, 3$>$, $S_{4})$ = \$7. Thus, the \textit{SRU} of $t$ in $S_{4}$ is \textit{SRU}($s, T, S_{4}$) = max\{\$12, \$7\} = \$12. Finally, the \textit{SRU} value of $s$ in $D_T$ can be calculated as follows: \textit{SRU}$(s,T)$ = \textit{SRU}$(s, T, S_{1})$ + \textit{SRU}$(s,T, S_{2})$ + \textit{SRU}$(s,T, S_{4})$ + \textit{SRU}$(s,T, S_{5})$ = \$12 + \$15 + \$12 + \$13 = \$52.

\begin{theorem}
	\label{SRU}
	Given a sequence $s'$, target sequence $T$, and $q$-sequence database $D$, if we suppose $s$ is the descendant of $s'$ that satisfies the condition $T \subset s$ in the \textit{LQS}-Tree, then the following can be obtained: $\textit{u}(s)\leq \textit{SRU}(s',T)$.
\end{theorem}

\begin{proof}
	\label{proof_SRU}
	It is assumed that $s$ can be obtained by concatenating two sequences $s'$ and $s''$, denoted as $s$ = $s' \ast s''$, where $s'$ is a prefix sequence of $s$ and $s''$ is the other nonempty part. Note that the symbol $''\ast''$ represents the concatenation of two sequences. Let us consider a $q$-sequence $S \in \textit{D}_T$ $\land s \subseteq S$; then, the utility of $s$ in $S$ is the sum of two parts, which is specified as $u(s,S)$ = $u(s', p, S)$ + $u_p(s'')$. $u(s', p, S)$ presents the utility of the instance of $s'$ within the extension position $p$ in $S$, while $u_p(s'')$ is the utility of a particular instance of $s''$ in $S$, whose first item is after $p$. It is not difficult to determine that $u_p(s'')$ $<$ $ \textit{ru}(s', p, s)$. Thus, we obtain the following.
	\begin{align*}
		u(s,S) &= u(s',\ p,\ S) + u_p(s'') \\
	&\leq u(s',\ p,\ S) + \textit{ru}(s',\ p,\ S) \\
	&\leq \mathop{max}\limits_{\forall p \in EP(s',S)}{u(s',\ p,\ s) + \textit{ru}(s',\ p,\ s)}. 
	\end{align*}
	Moreover, for each extension position $p$, it can be obtained that $\textit{Suf}(T,s')$ $\subseteq s''$ as $T \subseteq s$; then, the conclusion $\textit{Suf}(T,s')$  $\sqsubseteq$  $S/_{(s',\ p)}$ can be drawn, as $s'' $ $\sqsubseteq $ $S/_{(s',\ p)}$. Obviously, $\textit{ru}(s', p, S) > 0$ can be established because the second part of $s$, i.e., $s''$, is not empty.
	\begin{align*}
	u(s,S) &\leq \mathop{max}\limits_{\forall p \in EP(s',S)}{u(s',\ p,\ s) + \textit{ru}(s',\ p,\ s)} \\
	&= \mathop{max}\limits_{\forall p \in EP(s',S)}{\textit{SRU}(s',T,p,S)} \\
	&= \textit{SRU}(s',\ T,\ S).
	\end{align*}
	Subsequently, the following can be obtained: $u(s)$ = $\sum_{\forall S \in D}^{}u(s,S) \leq \sum_{\forall S \in D}^{}\textit{SRU}(s',\ T,\ S)$ = $\textit{SRU}(s',T)$.
\end{proof}

\begin{strategy}[depth pruning strategy]
	Let $s$ be a candidate that corresponds to the node $N$ in the \textit{LQS}-tree. If $\textit{SRU}(s,\ T)$ $< \xi $ $\times u(D_T)$, then each descendant node of $N$ can be pruned; thus, the TUSQ algorithm can stop at node $N$.
\end{strategy} 

\begin{definition}[terminated descendants utility]
	Given a sequence $s$, target sequence $T$, and $q$-sequence $S$, we assume that the sequence $\alpha$ can generate $s$ by one $I$-Extension or $S$-Extension. Note that \textit{Pre}($T,s)$ is the longest prefix of $T$ contained by $s$, and \textit{Suf}($T,s)$ is the suffix of $T$ with respect to $\textit{Pre}(T,s)$. The \textit{TDU} of $s$ in $S$ with respect to $T$, denoted as \textit{TDU}($s,T,S)$, is defined as:
	\[\textit{TDU}(s, T, S) = SRU(\alpha, T, S)\]
	if and only if \textit{s} $\sqsubseteq$ \textit{S} $\land$ $\alpha \sqsubseteq$ \textit{S} $\land$ \textit{Suf}$(T,s)$ $\sqsubseteq$ \textit{S}$/_{(s, pivot)}$ can be established; otherwise, \textit{TDU}$(s, T, S)$ = 0. The \textit{TDU} of $s$ with respect to $T$ in a $q$-sequence database $D$, denoted as $\textit{TDU}(s, T, D)$, is defined as:
	\[\textit{TDU}(s, T, D) = \sum_{\forall S \in D \land t \sqsubseteq S}^{}\textit{TDU}(s, T, S).\]
	
\end{definition} 

For instance, given sequences $s$ = $<$\{$a$\}, \{$cd$\}$>$ and $\alpha$ = $<$\{$a$\}, \{$c$\}$>$. It is obvious that $S_{1}$ and $S_{2}$ contain both $s$ and $\alpha$. We obtain \textit{TDU}$(s, T, S_{1})$ = \$0 because the pivot value of $t$ in $S_{1}$ is 3, and $S_{1}/_{(t, 3)}$ is an empty sequence that does not contain \textit{Suf}($T,s)$ = $<$\{$e$\}, \{$c$\}$>$. In addition, we have \textit{TDU}$(s, T, S_{2})$ = \textit{SRU}$(\alpha,T,S_{2})$ = \textit{SRU}$(\alpha$, $T$, $<$1, 2$>$, $S_{2})$ = \$18. Finally, the \textit{TDU} value of $t$ in $D$ is \textit{TDU}$(s, T)$ = \$18.

\begin{theorem}
	\label{TDU}
	Given a sequence $s'$, target sequence $T$, and $q$-sequence database $D$, if $s$ is the descendant of $s'$ or identical to $s'$, which satisfies $T \subset s$ in the \textit{LQS}-Tree, it can be obtained that $\textit{u}(s)\leq \textit{TDU}(s',T)$.
\end{theorem}

\begin{proof}
	\label{proof_TDU}
	Let us assume that $s'$ is the extension sequence of the sequence $\alpha$; thus, the node representing $\alpha$ is the parent node of that of $s'$ in \textit{LQS}-Tree. According to Proof \ref{proof_SRU}, given a $q$-sequence $S$ containing $s$, we obtain $u(s,S) \leq \textit{SRU}(\alpha, T, S)$ as $s' \subseteq$ $s \land T \subseteq s$. In addition, we can also obtain that \textit{Suf}($T,s') \sqsubseteq S/_{(s, p)}$ because $s'$ is a prefix of $s$ and $s \subseteq S$. In addition, $s' \subseteq S$ when $s \subseteq S$. Then, based on Definition \textit{SRU}, it can be observed that $u(s,S)$ $\leq$ $\textit{SRU}(\alpha, T, S)$ = $\textit{TDU}(s', T, S)$. Finally, it can be concluded that $u(s)$ = $\sum_{\forall S \in D}^{}u(s,S) \leq \sum_{\forall S \in D}^{}\textit{TDU}(t, T, S)$ = $\textit{TDU}(s',T)$.
\end{proof}

\begin{strategy}[width pruning strategy]
	Let $s$ be a candidate sequence corresponding to node $N$ in the \textit{LQS}-tree. If $\textit{TDU}(s,T)$ $< \xi $ $\times u(D_T)$, then the subtree rooted at $s$ can be pruned.
\end{strategy} 

\subsection{Proposed projected database}

This subsection introduces the developed data structure, named targeted chain, based on the concept of the utility chain \cite{wang2016efficiently}. As we know, a naive method to identify the utility of a sequence is to scan each $q$-sequence in the database. This method is inefficient as the number of candidates is too large and it is necessary to scan the entire database repeatedly under this method. To reduce the cost of scanning the database, several studies introduced compact database representations as well as the concept of a projected database to decrease the scale of databases when the algorithm explores longer candidates. The key idea of a projected database that is significantly smaller than the original database follows a simple principle; that is, a $q$-sequence $S$ does not contain a sequence $s$ if $S$ does not contain the prefix of $s$. Thus, the mining method scans the projected database of $s$, which is sufficient when calculating the utility of extension sequences of $s$. 

In the proposed algorithm, each $q$-sequence in the database is represented by a data structure called the $q$-matrix \cite{wang2016efficiently} for efficient access. The $q$-matrix of a $q$-sequence consists of a utility matrix and rest utility matrix, which record the utility of each item in the $q$-sequence and the utility of the corresponding rest sequence, respectively. To further illustrate this data structure, Figure \ref{qmatrix} shows an example of the $q$-matrix of $S_{1}$ as listed in Table \ref{table1}. More details about the $q$-matrix can be obtained in Ref. \cite{wang2016efficiently}.

\begin{figure}[htbp]
	\centering
	\includegraphics[width=0.8\linewidth]{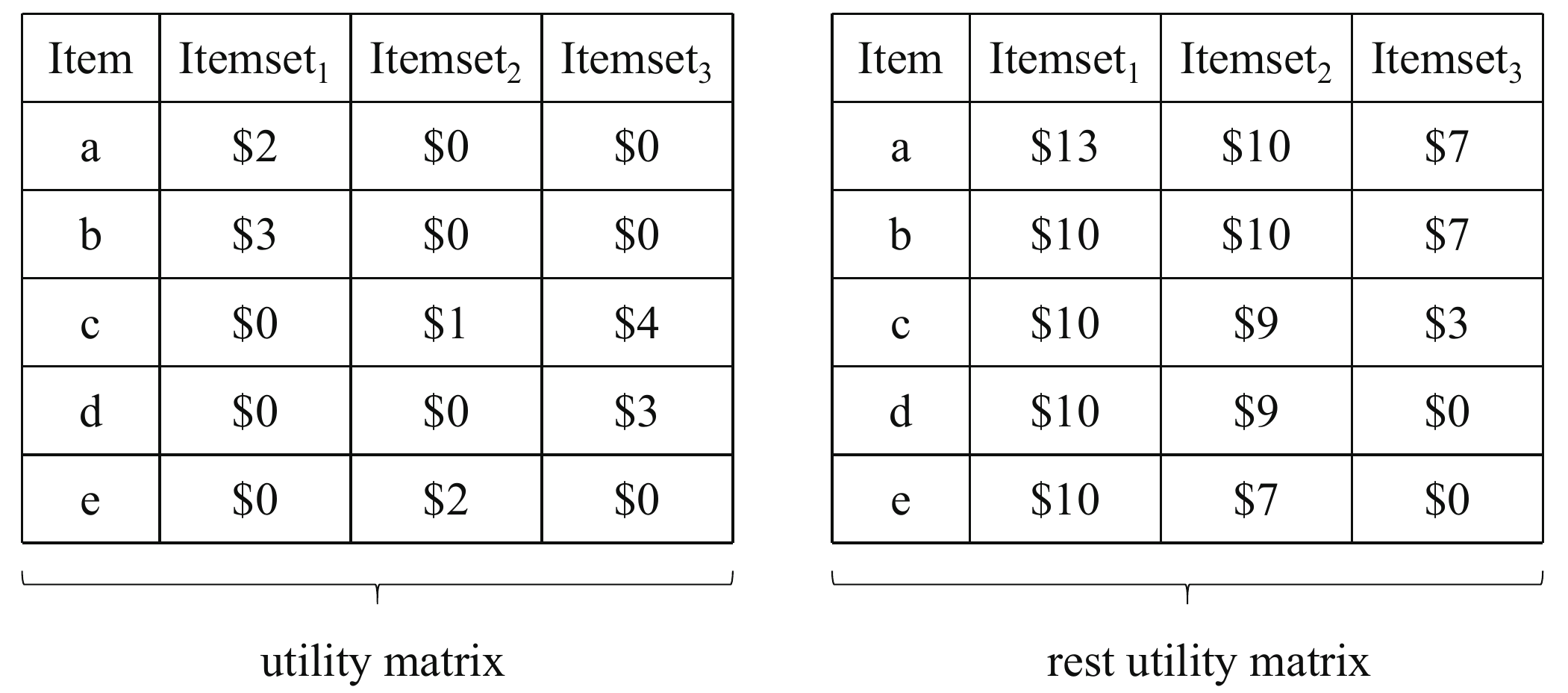}
	\caption{$q$-matrix of $S_{1}$ in Table \ref{table1}}
	\label{qmatrix}
\end{figure}

The developed targeted chain is a compact representation of the projected database consisting of essential information of $q$-sequences for utility and upper bound calculation. Figure \ref{targeted_chain} illustrates a targeted chain of the sequence $<$\{$a$\}, \{$c e$\}$>$ in the $q$-sequence database $D$, which is listed in Table \ref{table1}, as an example. As we can observe, a targeted chain consists of a head table and multiple targeted lists, where each row in the head table corresponds to a targeted list. To facilitate the description of the targeted chain, it is assumed that a sequence $s$ has multiple instances at $m$ promising extension positions $\textit{PEP}$ = $\{{pep}_{1}$, ${pep}_{2}$, $\cdots$, ${pep}_{m}\}$ in a $q$-sequence $S$; moreover, there are $n$ $q$-sequences, including $s$, where $s$ owns at least one promising extension position in the $q$-sequence database $D$. Then, the $n$ targeted lists corresponding to the $n$ $q$-sequences are arranged into the targeted chain. Moreover, the targeted list representing $S$ is a set of elements with a size of $m$, which corresponds to $m$ promising extension positions of $s$ in $S$. The $i$-th element contains the following fields: 1) \textit{TID}, which is the $i$-th promising extension position \({pep}_{i}\), 2) \textit{Utility}, which is the utility value of $s$ at the $i$-th promising extension position \({pep}_{i}\), and 3) \textit{RestUtility}, which is the rest utility of $s$ at the $i$-th promising extension position \({pep}_{i}\). The head table includes three key pieces of information: \textit{SID}, \textit{SRU}, and \textit{Prel}, which record the identifier of $S$, \textit{SRU} value of $s$ in $S$, and length of \textit{Pre}($T,t$), respectively.

\begin{figure}[htbp]
	\centering
	\includegraphics[width=0.82\linewidth]{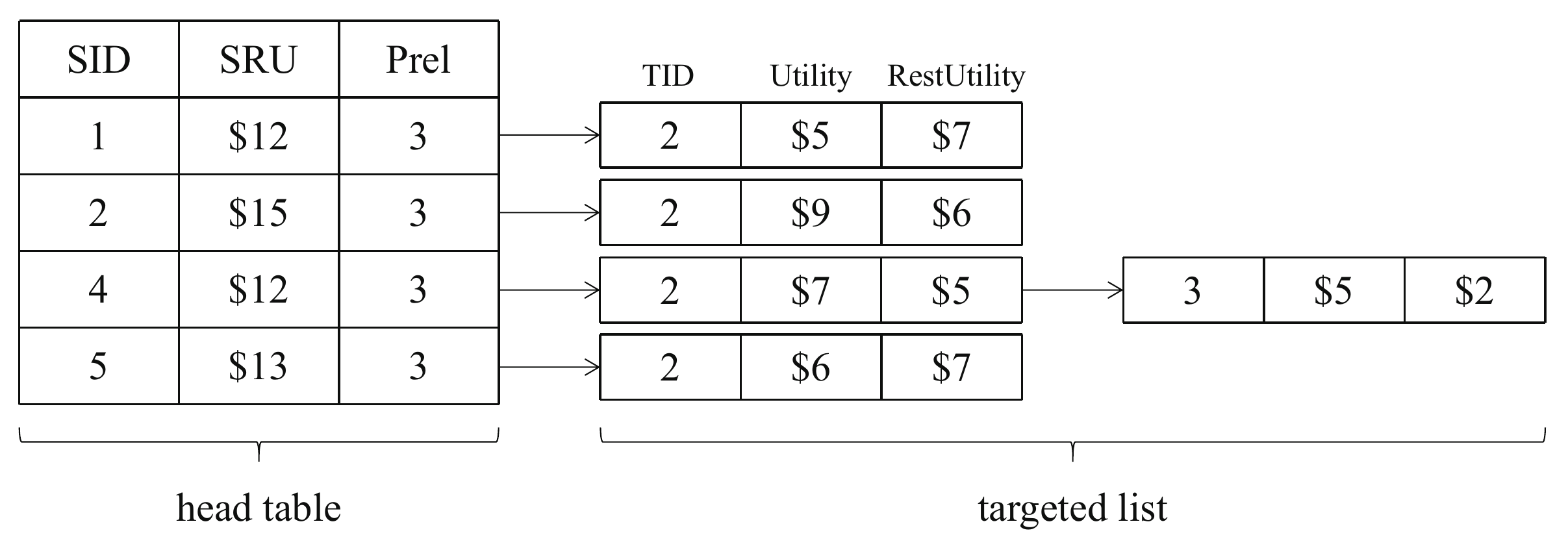}
	\caption{Targeted chain of sequence $<$\{$a$\},\{$ce$\}$>$}
	\label{targeted_chain}
\end{figure}

\subsection{Efficient method for calculating upper bounds}

In the developed upper bounds \textit{SRU} and \textit{TDU}, suffix information is utilized in the constraint conditions; that is, \textit{Suf(T,t)} $\sqsubseteq$ $S/_{(t,p)}$ in \textit{SRU}, and \textit{Suf(T,t)} $\sqsubseteq$ $S/_{(t,pivot)}$ in \textit{TDU}. It is time-consuming to determine whether one sequence is the subsequence of another as both sequences must be scanned from beginning to end, which results in low efficiency in the calculation of \textit{SRU} and \textit{TDU} values. To handle this problem, we developed a data structure called the \textit{LI}-Table to accelerate the checking process. 

\begin{definition}[last instance]
	Given a sequence $s$ = $<$$x_{1}$, $x_{2}$, $\cdots$, $x_{m}$$>$ and $q$-sequence $S$: $<$$X_{1}$, $X_{2}$, $\cdots$, $X_{n}$$>$, where $m \leq n$, we say that the last instance of $s$ in $S$ is located at position $p'$: $<$$k'_{1}$, $k'_{2}$, $\cdots$, $k'_{m}$$>$, if and only if $\forall p$: $<$$k_{1}$, $k_{2}$, $\cdots $, $k_{m}$$> \in$ $P(s,S)$ such that $k_{1} \leq k'_{1}$, $k_{2} \leq$ $k'_{2}$, $\cdots $, $k_{m}$ $\leq k'_{m}$.
\end{definition}

For example, $<$$\{a\}$,$ \{c\}$$>$ has the last instance in $S_{1}$ and $S_{2}$ at positions $p_{1}$: $<$1, 3$>$ and $p_{2}$: $<$3, 4$>$, respectively. Given a target sequence $T$ = $<$$z_{1}$, $z_{2}$, $\ldots$, $z_{n}$$>$ and $q$-sequence database $D$, the \textit{LI}-Table records the positions of the last instance of $T$ in each $q$-sequence of the filtered database $D_{T}$. The rows of the \textit{LI}-Table are indexed by the $q$-sequence identifier (i.e., $SID$), while the columns are indexed by the itemsets of $T$. Let there be a $q$-sequence $S$ = $<$$X_{1}$, $X_{2}$, $\ldots$, $X_{v}$$>$ with identifier $\textit{SID}_{S}$ in $D_{T}$, and let the position of the last instance of $T$ in $S$ is $p$: $<$$k'_{1}$, $k'_{2}$, $\ldots$, $k'_{n}$$>$. Then, the \textit{LI}-Table ($\textit{SID}_{S},i$) = $k'_{i} $ ($1\leq i\leq n, 1\leq k'_{i}\leq v)$. 

Let us consider the running example, $T$ = $<$\{\textit{a}\}, \{\textit{c} \textit{e}\}, \{\textit{c}\}$>$ has a last instance at position $<$1, 2, 4$>$ in $S_{2}$. Then, we obtain \textit{LI}(2,1) = 1, \textit{LI}(2,2) = 2, and \textit{LI}(2,3) = 4. The complete structure of the \textit{LI}-Table of $T$ in the running example $D$ is shown in Table \ref{LIT}.

\begin{table}[!htbp]
	\caption{Last instance table of $<$\{$a$\}, \{$ce$\}, \{$c$\}$>$ in the running example $D$}
	\label{LIT}
	\centering      
	\begin{tabular}{|c|c|c|c|}
		\hline
		\textbf{$SID$} & \textbf{$Itemset_{1}$} & \textbf{$Itemset_{2}$} & \textbf{$Itemset_{3}$} \\ \hline 
		$S_{1}$ & 1 & 2 & 3 \\ \hline
		$S_{2}$ & 1 & 2 & 4 \\ \hline
		$S_{4}$ & 1 & 3 & 4 \\ \hline
		$S_{5}$ & 1 & 2 & 3 \\ \hline
	\end{tabular}
\end{table}

Let us consider a sequence $s$ = $<$$x_{1}$, $x_{2}$, $\cdots$, $x_{m}$$>$. Suppose it has an instance, denoted as $t$, at position $<$$p_{1}$, $p_{2}$, $\cdots$, $p_{m}$$>$ in a $q$-sequence $S$. According to the aforementioned definitions, $S/_{(s,p_{m})}$ denotes the rest sequence of $S$ with respect to $p_{m}$. Moreover, $Suf(T,t)$ is the suffix of $T$ with respect to the longest prefix of $T$, that is, $Pre(T,t)$, contained by $s$. Let us assume that the first itemset in $\textit{Suf}(T,s)$, denoted as $Suf_{1}$, belongs to the $i$-th itemset of $T$. The value of $i$ can be obtained easily because the length of $Pre(T,s)$ is stored in the field \textit{Prel} of the targeted chain of $s$. We can determine whether $p_{m}$ is a promising extension position, that is, whether $Suf(T,t)$ $\sqsubseteq $ $S/_{(s,p_{m})}$, by the following deduction. If \textit{LI}$(SID_{S},i) $ $< p_{m}$, $\textit{Suf}(T,\ s)\sqsubseteq \textit{S}/_{(s,p_{m})}$ is not satisfied. Otherwise, $\textit{Suf}(T,\ s)\sqsubseteq \textit{S}/_{(s,p_{m})}$ can be established.

\begin{proof}
	\label{proof_LIT}
	Because $p_{m}$ is the extension position of $t$ in $S$, the position of the first itemset of $S/_{(t,p_{m})}$ in $S$ is $p_{m}$ or $p_{m}$+1. According to the definition of the \textit{LI}-Table, $Suf(T,t)$ has a last instance in $S$ at position $<$$\textit{LI}(SID_{S},i)$, $\textit{LI}(SID_{S}$, $i+1)$, $\cdots$, $\textit{LI}(SID_{S},n)$$>$. Because \textit{LI}$(SID_{S},i)<p_{m}$, $S/_{(t,p_{m})}$ does not contain $Suf_{1}$, \textit{Suf}$(T, t)$ does not have an instance in $S/_{(t,p_{m})}$. Thus, \textit{Suf}$(T, t)\sqsubseteq \textit{S}/_{(t,p_{m})}$ cannot be established.
\end{proof}

By introducing the store space-saving structure of the \textit{LI}-Table, the time-consuming subsequence-judging problem can be simplified as a simple numeric comparison, which significantly improves the efficiency of the calculation of \textit{SRU} and \textit{TDU} values.

\subsection{Proposed targeted querying algorithm}

Based on the above data structures and pruning strategies, the proposed TUSQ algorithm is described as follows. Algorithm \ref{alg:TUSQ} shows the main procedure of TUSQ, which takes a $q$-sequence database, an external utility table, a target sequence, and a minimum utility threshold as the inputs. After the mining process, TUSQ outputs the complete set of high-utility targeted queries. Initially, the utility values of $q$-items are calculated using the external utilities in \textit{UT}, which is an external utilities table, and internal utilities in $D$. At the beginning, TUSQ follows the \textit{DPP} strategy for filtering the redundant $q$-sequences when scanning $D$. At the same time, the utility value and \textit{SRU} upper bound of each 1-sequence existing in $D$ can be obtained. Note that several data structures, such as $q$-matrices, projected databases, and the \textit{LI}-Table, are constructed simultaneously in the database scan (Line 1). Then, TUSQ handles the 1-sequences. First, TUSQ identifies each 1-sequence according to Definition \ref{utility-driven targeted query} with the help of the targeted chain in its projected database (Lines 3--5). Then, TUSQ utilizes the depth pruning strategy to check each 1-sequence $s$. If the \textit{SRU} upper bound value of $s$ is not smaller than the threshold, it recursively calls the second procedure \textit{PatternGrowth} to extract longer UTQs with the prefix $s$ (Lines 6--8) for traversing the \textit{LQS}-tree. Finally, it can return all the desired queries (Line 10).

\begin{algorithm}[htbp]
	\caption{TUSQ Algorithm}
	\label{alg:TUSQ}
	\begin{algorithmic}[1]
		\REQUIRE 
		$D$: a $q$-sequence database;	\textit{UT}: external utility table containing the external utility of each item; $T$: a target sequence;
		$\xi$: a minimum utility threshold.
		\ENSURE 
		\textit{UTQ}: the complete set of utility-driven targeted queries;
		
		\STATE scan $D$ to: \\
		i) filter out redundant sequences and construct the filtered database $D_T$; $//$ The \textit{DPP} strategy\\
		ii) calculate (1) the utility value and \textit{SRU} of each 1-sequence in $D_T$, (2) the utility of $D_T$, i.e., $u(D_T)$;\\
		iii) construct (1) the $q$-matrix of each $q$-sequence in $D_T$, (2) projected databases of all 1-sequences, (3) \textit{LI-Table} based on $T$ and $D_T$;
		
		\FOR {each $s \in $ 1-sequences }
		\IF{$u(s) \ge \xi\times u(D_{T})\land$ $s.\textit{Prel}$ = $T.\textit{length}$}
		\STATE update $\textit{UTQ} \leftarrow \textit{UTQ} \cup s$;
		\ENDIF \\
		$//$ Depth pruning strategy
		\IF{$\textit{SRU}(s,T) \geq \xi\times u(D_T)$}
		\STATE call \textit{PatternGrowth($s$, $s$.\textit{ProjectedDatabase})};
		\ENDIF
		\ENDFOR
		
		\RETURN \textit{UTQ};
	\end{algorithmic}	
\end{algorithm}

Algorithm \ref{alg:PG} presents the details of the recursively searching procedure. It takes a sequence $s$ to be the prefix and the corresponding projected database $s.\textit{ProjectedDatabase}$ as input. As the first step, it obtains the $q$-matrix \textit{qm} corresponding to each targeted list in $s.\textit{ProjectedDatabase}$. By searching in \textit{qm}, it adds the items that can be used for $I$-Extension and $S$-Extension of $s$ into \textit{ilist} and \textit{slist}, respectively (Lines 1--5). Then, it handles each item $i$ in \textit{ilist} (Lines 6--14). The collected item $i$ is extended to $s$ to generate the extension sequence $s'$ (Line 7). Then, the procedure checks the candidate $s'$ using the \textit{TDU} value with respect to the target $T$ based on the width pruning strategy (Lines 8--10). Note that the \textit{TDU} value of $s'$ is calculated when searching \textit{qm} (Line 3). If $s'$ is not pruned, then the projected database of $s'$ is built based on $s.\textit{ProjectedDatabase}$ and $s'$ is placed in the set \textit{PromisingSeq} (Lines 21--22). The other set of items for the S-Extension \textit{slist} can be handled similarly (Lines 15--23). The newly generated sequences are first evaluated to determine whether they are UTQs (Lines 25--27). Following the depth pruning strategy, the procedure invokes itself recursively with $s'$ for examining the new and longer sequence to go deeper in the \textit{LQS}-Tree (Lines 28--30).

In addition, considering that there is no previous algorithm for the task of targeted utility-oriented sequence querying, a baseline method named HUS-UTQ, which can also mine UTQs from a $q$-sequence database, is designed based on the HUS-Span algorithm \cite{wang2016efficiently}. The structure of HUS-UTQ is described briefly as follows. First, the \textit{DPP} strategy is performed on the input $q$-sequence database $D$ to obtain the filtered database $D_{T}$. Then, HUS-Span is utilized to mine the complete set of HUSPs from $D_{T}$. Finally, the UTQs are derived by selecting the patterns containing the target sequence from the HUSPs. When compared to the well-designed TUSQ, HUS-UTQ is simpler but rougher. The primary difference between these two algorithms lies in the upper bounds they use. Moreover, TUSQ adopts the efficient \textit{LI}-Table to accelerate the calculation of upper bounds. 

\begin{algorithm}[htbp]
	\caption{PatternGrowth}
	\label{alg:PG}
	\begin{algorithmic}[1]
		\REQUIRE 
		$s$: a sequence as the prefix; \\
		$s$.\textit{ProjectedDatabase}: the projected database of the sequence $s$.
		
		\FOR {each targeted list $\textit{tl} \in s$.\textit{ProjectedDatabase}}
		\STATE get the $q$-matrix \textit{qm} corresponding to \textit{tl}
		\STATE search \textit{qm} to put $I$-Extension items into \textit{ilist};
		\STATE search \textit{qm} to put $S$-Extension items into \textit{slist};
		\ENDFOR
		
		\FOR {each item $i \in \textit{ilist}$ }
		\STATE $s'$ $\leftarrow$ $<$ $s \oplus i$$>$; \\
		$//$ Width pruning strategy
		\IF{$\textit{TDU}(s',T)$ $< \xi \times$ $u(D_T)$}
		\STATE remove $i$ from \textit{ilist};
		\STATE continue;
		\ENDIF
		\STATE build projected database of $s'$ $\textit{ProjectedDatabase}_{s'}$; 
		\STATE put $s'$ into \textit{PromisingSeq};
		\ENDFOR
		
		\FOR {each item $i \in \textit{slist}$ }
		\STATE $s'$ $\leftarrow$ $<$ $s \otimes i$$>$; \\
		$//$ Width pruning strategy
		\IF{$\textit{TDU}(s',T)$ $< \xi \times$ $u(D_T)$}
		\STATE remove $i$ from \textit{slist};
		\STATE continue;
		\ENDIF
		\STATE construct $s'$.\textit{ProjectedDatabase}; 
		\STATE put $s'$ into \textit{PromisingSeq};
		\ENDFOR
		
		\FOR {each sequence $\textit{s'} \in$   \textit{PromisingSeq}}
		\IF{$u(s') \geq u(D_T) \times \xi \land s'.\textit{Prel}$ = $T.\textit{length}$}
		\STATE update $\textit{UTQ} \leftarrow \textit{UTQ} \cup s'$;
		\ENDIF \\
		$//$ Depth pruning strategy
		\IF{$\textit{SRU}(s',T)\geq \xi\times u(D_T)$}
		\STATE call \textit{PatternGrowth($s'$, $s'$.\textit{ProjectedDatabase})};
		\ENDIF
		\ENDFOR
	\end{algorithmic}	
\end{algorithm}

\section{Experiments}   \label{sec:experiments}

To measure the effectiveness of the developed TUSQ algorithm, substantial experiments were conducted. All experiments were performed on a personal workstation equipped with 3.40 GHz Intel Core i7-10700K Processor with 32.0 GB of main memory to run the Windows 10 64-bit operating system. All algorithms executed in the experiments were implemented in Java using the IntelliJ IDEA compiler in Community Edition.

\begin{table}[!ht]
	\caption{Features of the datasets}
	\label{features}
	\centering      
	\begin{tabular}{|c|c|c|c|c|c|}
		\hline
		\textbf{Dataset} & \textbf{$|\textit{D}|$} & \textbf{$|\textit{I}|$} & \textbf{$\textit{avg}(\textit{S})$} & \textbf{$\textit{max}(\textit{S})$} &  \textbf{$\textit{avg}(\textit{Its})$} \\ \hline 
		\textit{MSNBC} & 31,790 & 17 & 13.33 & 100 & 1.00 \\ \hline     
		\textit{Yoochoose} & 234,300 & 16,004 & 2.25 & 112 & 1.97 \\ \hline
		\textit{Syn40K} & 40000 & 7584 & 6.19 & 18 & 4.32 \\ \hline
		\textit{Bible} & 36,369 & 13,905 & 21.64 & 100 & 1.00 \\ \hline
		\textit{Leviathan} & 5,834 & 9,025 & 33.81 & 100 & 1.00 \\ \hline
		\textit{Sign} & 730 & 267 & 52.00 & 94 & 1.00 \\ \hline
	\end{tabular}
\end{table}

\subsection{Data description}

In the experiments, we used five real-world datasets and a synthetic dataset to evaluate the performance of the algorithms. Among these datasets, \textit{MSNBC} and \textit{Yoochoose} are composed of web clickstream data from msnbc.com and an e-commerce website, respectively. \textit{Syn40K} is a synthetic dataset generated by IBM data generator \cite{agrawal1994fast}. Moreover, \textit{Bible} and \textit{Leviathan} are conversions of the Bible and the novel Leviathan, respectively, where each word is converted into an item and each sentence is regarded as a sequence. Further, \textit{Sign} is a sign language utterance dataset. Note that \textit{Yoochoose} can be obtained from the website of ACM Recomender Systems\footnote{https://recsys.acm.org/recsys15/challenge/}, and the other five datasets are available at SPMF, an open-source data mining library\footnote{http://www.philippe-fournier-viger.com/spmf/}. The above six datasets have different characteristics, which can represent a majority of data types obtained in real-life situations. The features of the six datasets are listed in Table \ref{features}. Note that $|D|$ is the number of $q$-sequences, $|I|$ is the number of distinct $q$-items, $\textit{avg}(S)$ and $\textit{max}(S)$ are the average and maximum length of $q$-sequences, respectively, $\textit{avg}(Its)$ is the average number of $q$-items per $q$-itemsets.

\begin{figure*}[htbp]
	\centering
	\includegraphics[height=0.3\textheight,width=0.95\linewidth]{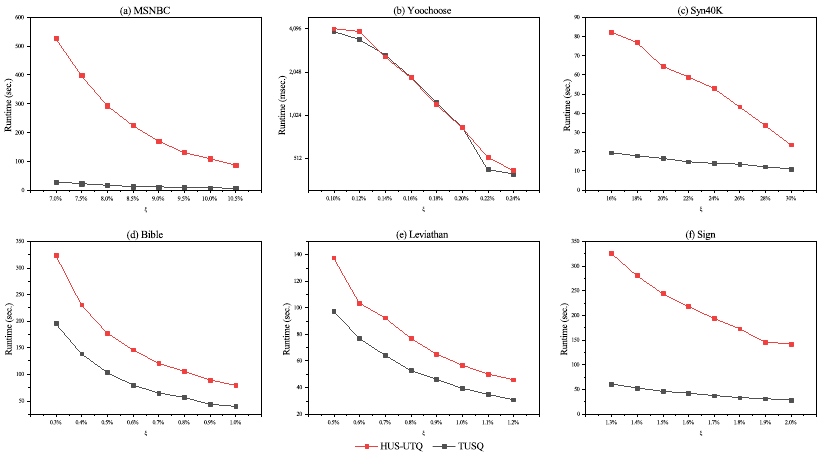}
	\caption{Runtime of the compared methods under various minimum utility thresholds}
	\label{runtime}
\end{figure*}

\begin{figure*}[htbp]
	\centering
	\includegraphics[height=0.3\textheight,width=0.95\linewidth]{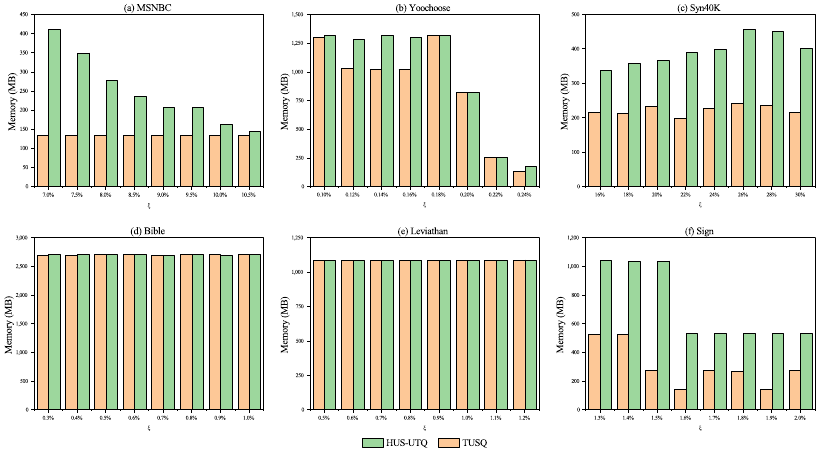}
	\caption{Memory of the compared methods under various minimum utility thresholds}
	\label{memory}
\end{figure*}

\begin{figure*}[htbp]
	\centering
	\includegraphics[height=0.33\textheight,width=0.95\linewidth]{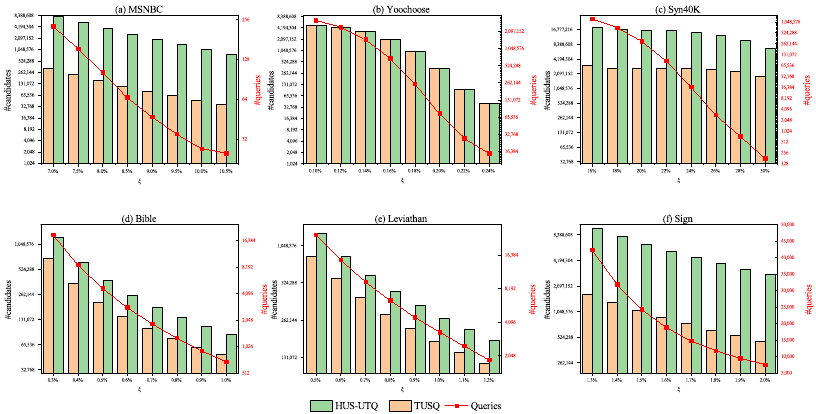}
	\caption{Candidates of the compared methods under various minimum utility thresholds}
	\label{candidates}
\end{figure*}

\subsection{Efficiency}

This subsection presents a performance comparison of the datasets mentioned above in terms of runtime and memory consumption. While considering the execution time of the proposed TUSQ method, Figure \ref{runtime} shows that the TUSQ outperformed the baseline HUS-UTQ in all cases. However, the advantages of TUSQ are not obvious in the \textit{Yoochoose} dataset. We speculate that this is because the $\textit{avg}(S)$ of \textit{Yoochoose} dataset is small, both algrithms can calculate the utility and upper bound values of candidates very quickly. Moreover, although the difference between the running time of the two algorithms is relatively large in the other five datasets, smaller improvement can be observed with the increase in the minimum utility threshold. In theory, the search space becomes increasingly large scale when the threshold is decreased. This is because the developed pruning strategies cannot prune certain branches of the \textit{LQS}-Tree that correspond to a small threshold, which implies that the method must traverse an enormous scope. This satisfactorily explains the decrease in runtime with increasing threshold. 

In addition to the runtime, another key measurement criterion of performance in the domain of data mining is the memory usage, which is compared and illustrated in Figure \ref{memory}. It is clear that the TUSQ occupies less memory than the baseline approach with all the parameter settings on the \textit{MSNBC}, \textit{Syn40K}, and \textit{Sign} datasets. However, the memory consumed by TUSQ was almost equal to that of HUS-UTQ in the \textit{Bible} and \textit{Leviathan} datasets, and the memory usage was stable under various threshold settings. In \textit{Yoochoose}, it can be observed that the TUSQ outperformed HUS-UTQ in certain cases, such as when the threshold was set to 0.16\%; however, the two methods sometimes consumed nearly the same memory for storing built structures. 

In summary, in most cases on the six datasets, the TUSQ algorithm significantly outperformed the HUS-UTQ algorithm in terms of runtime and memory usage. This is owing to the use of the compact data structure target chain and two pruning strategies in TUSQ, which significantly reduced the time and space complexity.

\subsection{Number of candidates and queries}

This subsection presents the number of candidates that intuitively demonstrates the size of the actual search space, generated by two compared algorithms. Moreover, the number of desired targeted queries is also shown in Figure \ref{candidates}. Note that $\textit{\#candidates}$ and $\textit{\#queries}$ represent the number of candidates and the number of UTQs, respectively. As can be observed in Figure \ref{candidates}, the set of candidates to be identified in the mining process of TUSQ was always significantly more than that of the baseline under various minimum utility thresholds. This is because two novel upper bounds that tightly overestimate utility values are adopted in TUSQ; consequently, the corresponding pruning strategies are able to eliminate unpromising candidates in an early stage. It is not difficult to understand the fact that the number of candidates decreases with an increase in the minimum utility threshold because the pruning strategies terminate enumerating those candidates from a short length when the threshold is large. In addition, the number of desired UTQs also follows the equable descent but is significantly smaller than the number of candidates in the same parameter setting. These phenomena reflect the fact that a large number of candidates are generated in the mining process; however, only a few of them are the final desired targeted queries. In conclusion, the designed TUSQ algorithm can easily check the data as less candidates as possible and retrieve the UTQs efficiently.

\begin{figure*}[!htb]
	\centering
	\includegraphics[height=0.56\textheight,width=0.95\linewidth]{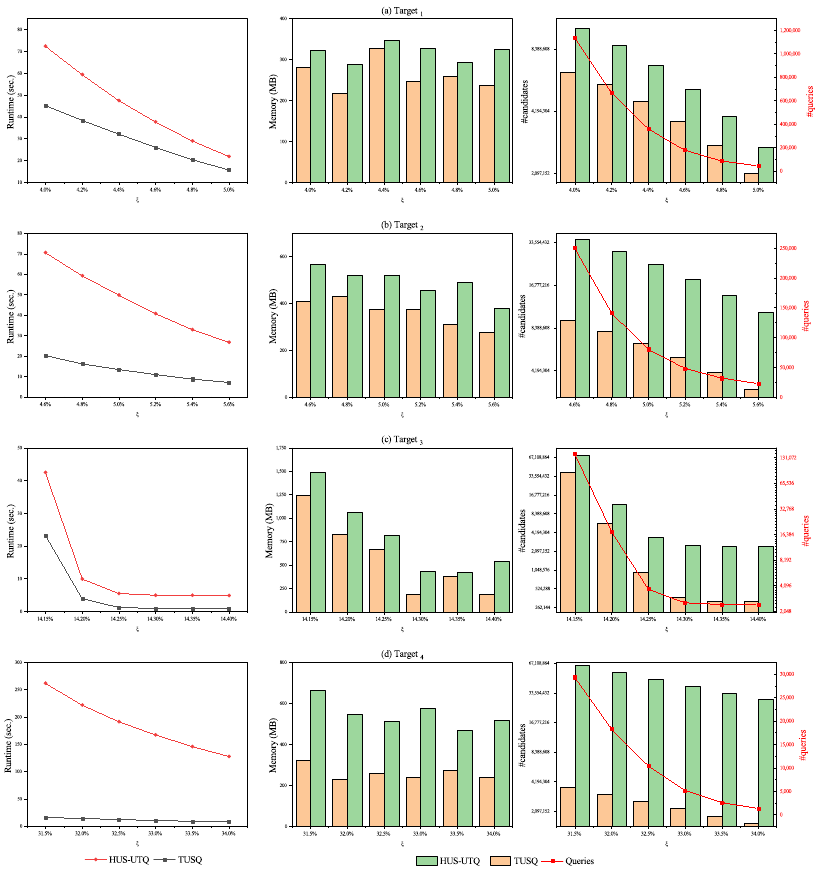}
	\caption{Performance of the compared methods under four target sequences and various minimum utility thresholds}
	\label{casestudy}
\end{figure*}

\subsection{Target sequence analysis}

In this subsection, we present the results of a series of experiments that were conducted on the \textit{Syn80K} dataset to evaluate the performance of the proposed method with various types of target sequences. Let there be four target sequences $\rm Target_1$: $<$$\{\textit{8636}\}$$>$, $\rm Target_2$: $<$\{\textit{2512}\ \textit{3180}\ \textit{5894}\}$>$, $\rm Target_3$: $<$\{\textit{2058}\}, \{\textit{7544}\}, \{\textit{8504}\}$>$, and $\rm Target_4$: $<$\{\textit{8856}\}, \{\textit{2058}\ \textit{3118}\ \textit{8389}\}, \{\textit{2276}\ \textit{5116}\}, \{\textit{2816}\}$>$, which represent a single-item-based sequence, an itemset-based sequence, a multiple-item-based sequence, and a hybrid sequence, respectively. The experimental results are shown in Figure \ref{casestudy}. As it can be clearly observed, Figure \ref{casestudy} not only demonstrates that the TUSQ has a better performance in terms of runtime, memory usage, and candidate filtering, but it is also skilled in addressing the database with respect to a complex target sequence. For example, the memory usage consumed by HUS-UTQ was more than that of TUSQ in all cases. The same phenomenon can be observed in the figures illustrating the performance in terms of runtime and candidate filtering. Moreover, it is interesting to observe that with the increase in the minimum utility threshold, the difference between the two compared algorithms becomes increasingly wider, especially with itemset-based sequences and hybrid sequences. This illustrates that our proposed algorithm is good at addressing the target sequences with a large $|I|$ value under a large minimum utility threshold.

\subsection{Evaluation of the proposed pruning strategies}

The performance of the methods without the designed pruning strategies was researched to evaluate the effect of the pruning strategies. Under the same threshold settings, a series of experiments were conducted, whose results are shown in Figure \ref{effectiveness}.

\begin{figure}[!htb]
	\centering
	\includegraphics[width=0.9\linewidth]{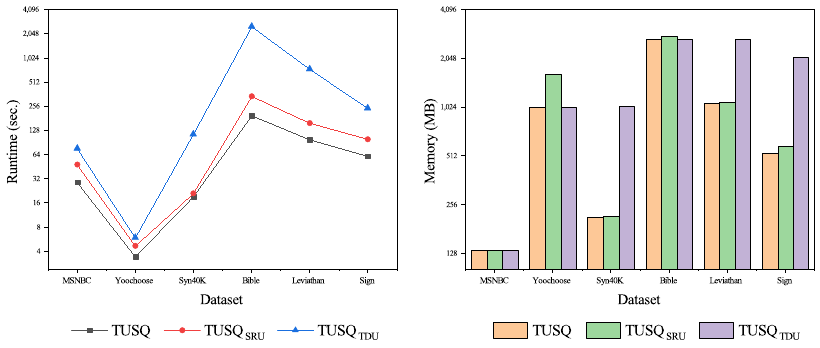}
	\caption{Effectiveness of the proposed pruning strategies}
	\label{effectiveness}
\end{figure}

 Note that TUSQ$_{\rm SRU}$ and TUSQ$_{\rm TDU}$ are the methods ablating the \textit{SRU} and \textit{TDU} upper bounds from the proposed TUSQ algorithm, respectively. To control the variate in the experiments, we set the minimum utility threshold to 7.0, 0.12, 16.0, 0.3, 0.5, and 1.3\% when executing the three algorithms on the six datasets. It can be observed that TUSQ outperforms the two variants, especially in terms of running time. In particular, the two variants require more runtime than TUSQ in all six cases. For example, TUSQ performs the mining process on the \textit{Bible} dataset in 195 s; however, the two variants require 340 and 2516 s, respectively. While considering the memory consumption, the performance of TUSQ is superior; however, the difference in performance is less obvious. For instance, in the \textit{MSNBC} dataset, the three approaches require almost the same size of memory for storing the projected database. Another interesting fact is that the variant TUSQ$_{\rm TDU}$ requires more time to discover the desired queries than TUSQ$_{\rm SRU}$ in all cases, which implies that \textit{TDU} is more powerful for pruning the search space than \textit{SRU}. In conclusion, the results shown in Figure \ref{effectiveness} present the positive effect of the proposed pruning strategies in TUSQ for mining the complete set of UTQs efficiently.

\subsection{Effectiveness of the mining results}

The primary contribution of this study is the proposed target-oriented utility sequence querying task. Finally, we conducted a series of experiments utilizing the conventional HUSPM method HUS-Span to compare the target-oriented utility sequence querying with HUSPM on the four datasets with various parameter settings. As it can be observed in Figure \ref{patterns}, for each dataset, the number of targeted queries discovered by TUSQ was significantly less than that of high-utility queries extracted by HUS-Span with the same minimum utility threshold. For example, in the \textit{MSNBC} dataset, the number of queries was reduced by a thousand times when adding the constraint of the target sequence. In the less ideal case, for example, in the \textit{Leviathan} dataset, under the same threshold, HUS-Span could mine more queries in one time than the targeted queries of TUSQ. In conclusion, the proposed TUSQ algorithm is a practical and impeccable technology to interact with users by information selection, that is, querying with their targets. This ability significantly reduces the number of returned queries, which is beneficial for filtering redundant and useless information for users.

\begin{figure}[!htb]
	\centering
	\includegraphics[width=0.9\linewidth]{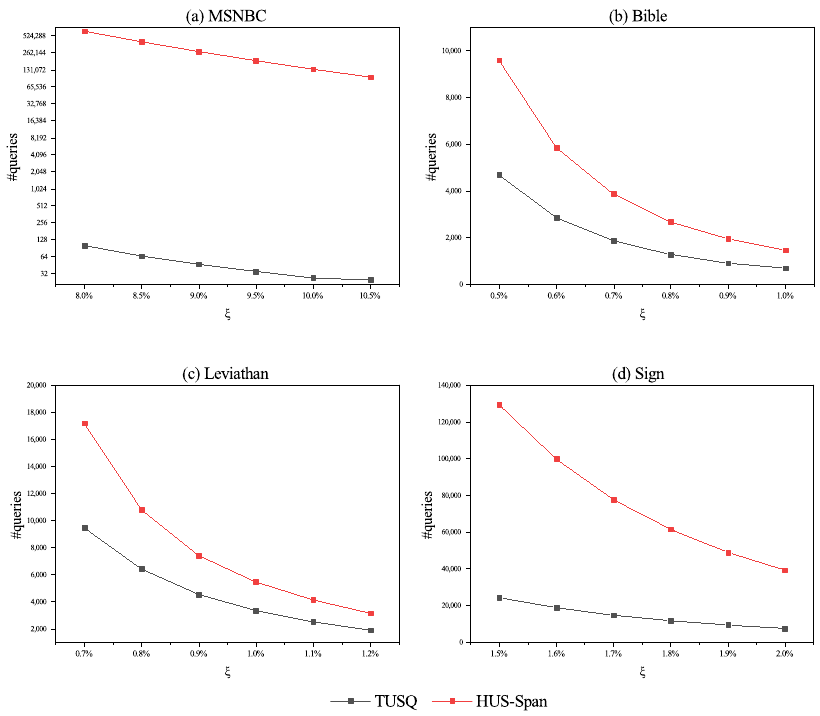}
	\caption{Number of queries discovered by TUSQ and HUS-Span}
	\label{patterns}
\end{figure}

\section{Conclusion}   \label{sec:conclusion}

FIM, SPM, ARM, HUIM, and HUSPM were emerging research topics in the domain of KDD. These data querying technologies aim to discover desired patterns, as ordinary queries, from the database with different measurements. In recent years, DBMS has been widely applied for managing an enormous collection of multisource, heterogeneous, complex, or growing data. A practical and impeccable DBMS can interact with users by information selection, that is, querying with their targets. However, the aforementioned querying algorithms cannot handle this problem. In this study, we incorporated utility into target-oriented SPQ and formulated the task of targeted utility-oriented sequence querying. To address the proposed problem, we developed a novel algorithm, namely TUSQ, which relies on a projection technology utilizing a compact data structure called the targeted chain. For further efficiency, two novel upper bounds \textit{SRU} and \textit{TDU}, as well as two pruning strategies were designed. Moreover, a structure named LI-Table was introduced to facilitate the calculation of upper bounds. An extensive experimental study conducted on several real and synthetic datasets shows that the proposed algorithm outperformed the designed baseline algorithm in terms of runtime, memory consumption, and candidate filtering. When compared to the HUSPs with lots of redundancy, as discovered using the HUSPM methods, the targeted queries that were returned by TUSQ are more interesting and useful for users.

%%%%%%%%%%%%%%%%%%%%%%%%%%%%%%%%%%%%%%%

% if have a single appendix:
%\appendix[Proof of the Zonklar Equations]
% or
%\appendix  % for no appendix heading
% do not use \section anymore after \appendix, only \section*
% is possibly needed

% use appendices with more than one appendix
% then use \section to start each appendix
% you must declare a \section before using any
% \subsection or using \label (\appendices by itself
% starts a section numbered zero.)
%

% use section* for acknowledgment
\ifCLASSOPTIONcompsoc
  % The Computer Society usually uses the plural form
  \section*{Acknowledgments}
\else
  % regular IEEE prefers the singular form
  \section*{Acknowledgment}
\fi

This work was supported by National Natural Science Foundation of China (Grant No. 62002136), Natural Science Foundation of Guangdong Province, China (Grant NO. 2020A1515010970) and Shenzhen Research Council (Grant NO. JCYJ20200109113427092, GJHZ20180928155209705).

% Can use something like this to put references on a page
% by themselves when using endfloat and the captionsoff option.
\ifCLASSOPTIONcaptionsoff
  \newpage
\fi

\bibliographystyle{IEEEtran}
% argument is your BibTeX string definitions and bibliography database(s)
\bibliography{TUSQ.bib}

% Generated by IEEEtran.bst, version: 1.14 (2015/08/26)
\begin{thebibliography}{10}
\providecommand{\url}[1]{#1}
\csname url@samestyle\endcsname
\providecommand{\newblock}{\relax}
\providecommand{\bibinfo}[2]{#2}
\providecommand{\BIBentrySTDinterwordspacing}{\spaceskip=0pt\relax}
\providecommand{\BIBentryALTinterwordstretchfactor}{4}
\providecommand{\BIBentryALTinterwordspacing}{\spaceskip=\fontdimen2\font plus
\BIBentryALTinterwordstretchfactor\fontdimen3\font minus
  \fontdimen4\font\relax}
\providecommand{\BIBforeignlanguage}[2]{{%
\expandafter\ifx\csname l@#1\endcsname\relax
\typeout{** WARNING: IEEEtran.bst: No hyphenation pattern has been}%
\typeout{** loaded for the language `#1'. Using the pattern for}%
\typeout{** the default language instead.}%
\else
\language=\csname l@#1\endcsname
\fi
#2}}
\providecommand{\BIBdecl}{\relax}
\BIBdecl

\bibitem{agrawal1994fast}
R.~Agrawal, R.~Srikant \emph{et~al.}, ``Fast algorithms for mining association
  rules,'' in \emph{Proceedings of the 20th International Conference on Very
  Large Data Bases}.\hskip 1em plus 0.5em minus 0.4em\relax Morgan Kaufmann,
  1994, pp. 487--499.

\bibitem{hipp2000algorithms}
J.~Hipp, U.~G{\"u}ntzer, and G.~Nakhaeizadeh, ``Algorithms for association rule
  mining-a general survey and comparison,'' \emph{ACM SIGKDD Explorations
  Newsletter}, vol.~2, no.~1, pp. 58--64, 2000.

\bibitem{zaki2001spade}
M.~J. Zaki, ``{SPADE}: An efficient algorithm for mining frequent sequences,''
  \emph{Machine Learning}, vol.~42, no. 1-2, pp. 31--60, 2001.

\bibitem{fournier2017survey}
P.~Fournier-Viger, J.~C.~W. Lin, B.~Vo, T.~T. Chi, J.~Zhang, and H.~B. Le, ``A
  survey of itemset mining,'' \emph{Wiley Interdisciplinary Reviews: Data
  Mining and Knowledge Discovery}, vol.~7, no.~4, p. e1207, 2017.

\bibitem{agrawal1995mining}
R.~Agrawal and R.~Srikant, ``Mining sequential patterns,'' in \emph{Proceedings
  of the 11th International Conference on Data Engineering}.\hskip 1em plus
  0.5em minus 0.4em\relax IEEE, 1995, pp. 3--14.

\bibitem{ayres2002sequential}
J.~Ayres, J.~Flannick, J.~Gehrke, and T.~Yiu, ``Sequential pattern mining using
  a bitmap representation,'' in \emph{Proceedings of the eighth ACM SIGKDD
  International Conference on Knowledge Discovery and Data Mining}.\hskip 1em
  plus 0.5em minus 0.4em\relax {ACM}, 2002, pp. 429--435.

\bibitem{liu2005two}
Y.~Liu, W.-k. Liao, and A.~Choudhary, ``A two-phase algorithm for fast
  discovery of high utility itemsets,'' in \emph{Proceedings of the 9th
  Pacific-Asia Conference on Knowledge Discovery and Data Mining}.\hskip 1em
  plus 0.5em minus 0.4em\relax Springer, 2005, pp. 689--695.

\bibitem{fournier2014fhm}
P.~Fournier-Viger, C.~W. Wu, S.~Zida, and V.~S. Tseng, ``{FHM}: Faster
  high-utility itemset mining using estimated utility co-occurrence pruning,''
  in \emph{Proceedings of the 21st International Symposium on Methodologies for
  Intelligent Systems}.\hskip 1em plus 0.5em minus 0.4em\relax Springer, 2014,
  pp. 83--92.

\bibitem{ahmed2010novel}
C.~F. Ahmed, S.~K. Tanbeer, and B.~S. Jeong, ``A novel approach for mining
  high-utility sequential patterns in sequence databases,'' \emph{ETRI
  Journal}, vol.~32, no.~5, pp. 676--686, 2010.

\bibitem{shie2011mining}
B.~E. Shie, H.~F. Hsiao, V.~S. Tseng, and P.~S. Yu, ``Mining high utility
  mobile sequential patterns in mobile commerce environments,'' in
  \emph{Proceedings of the 16th International Conference on Database Systems
  for Advanced Applications}.\hskip 1em plus 0.5em minus 0.4em\relax Springer,
  2011, pp. 224--238.

\bibitem{gan2021survey}
W.~Gan, J.~C.~W. Lin, P.~Fournier-Viger, H.~C. Chao, V.~S. Tseng, and P.~S. Yu,
  ``A survey of utility-oriented pattern mining,'' \emph{IEEE Transactions on
  Knowledge and Data Engineering}, vol.~33, no.~4, pp. 1306--1327, 2021.

\bibitem{mccarthy1989architecture}
D.~McCarthy and U.~Dayal, ``The architecture of an active database management
  system,'' \emph{ACM SIGMOD Record}, vol.~18, no.~2, pp. 215--224, 1989.

\bibitem{stonebraker1991postgres}
M.~Stonebraker and G.~Kemnitz, ``The postgres next generation database
  management system,'' \emph{Communications of the ACM}, vol.~34, no.~10, pp.
  78--92, 1991.

\bibitem{shabtay2018guided}
L.~Shabtay, R.~Yaari, and I.~Dattner, ``A guided fp-growth algorithm for
  multitude-targeted mining of big data,'' \emph{arXiv preprint
  arXiv:1803.06632}, 2018.

\bibitem{kubat2003itemset}
M.~Kubat, A.~Hafez, V.~V. Raghavan, J.~R. Lekkala, and W.~K. Chen, ``Itemset
  trees for targeted association querying,'' \emph{IEEE Transactions on
  Knowledge and Data Engineering}, vol.~15, no.~6, pp. 1522--1534, 2003.

\bibitem{fournier2013meit}
P.~Fournier-Viger, E.~Mwamikazi, T.~Gueniche, and U.~Faghihi, ``{MEIT}: Memory
  efficient itemset tree for targeted association rule mining,'' in
  \emph{Proceedings of International Conference on Advanced Data Mining and
  Applications}.\hskip 1em plus 0.5em minus 0.4em\relax Springer, 2013, pp.
  95--106.

\bibitem{abeysinghe2017query}
R.~Abeysinghe and L.~Cui, ``Query-constraint-based association rule mining from
  diverse clinical datasets in the national sleep research resource,'' in
  \emph{Proceedings of International Conference on Bioinformatics and
  Biomedicine}.\hskip 1em plus 0.5em minus 0.4em\relax IEEE, 2017, pp.
  1238--1241.

\bibitem{chiang2003goal}
D.~A. Chiang, Y.~F. Wang, S.~L. Lee, and C.~J. Lin, ``Goal-oriented sequential
  pattern for network banking churn analysis,'' \emph{Expert Systems with
  Applications}, vol.~25, no.~3, pp. 293--302, 2003.

\bibitem{chueh2010mining}
H.~E. Chueh \emph{et~al.}, ``Mining target-oriented sequential patterns with
  time-intervals,'' \emph{International Journal of Computer Science \&
  Information Technology}, vol.~2, no.~4, pp. 113--123, 2010.

\bibitem{chand2012target}
C.~Chand, A.~Thakkar, and A.~Ganatra, ``Target oriented sequential pattern
  mining using recency and monetary constraints,'' \emph{International Journal
  of Computer Applications}, vol.~45, no.~10, 2012.

\bibitem{abeysinghe2018query}
R.~Abeysinghe and L.~Cui, ``Query-constraint-based mining of association rules
  for exploratory analysis of clinical datasets in the national sleep research
  resource,'' \emph{BMC Medical Informatics and Decision Making}, vol.~18,
  no.~2, p.~58, 2018.

\bibitem{li2008fast}
H.~F. Li, H.~Y. Huang, Y.~C. Chen, Y.~J. Liu, and S.~Y. Lee, ``Fast and memory
  efficient mining of high utility itemsets in data streams,'' in
  \emph{Proceedings of the 8th International Conference on Data Mining}.\hskip
  1em plus 0.5em minus 0.4em\relax IEEE, 2008, pp. 881--886.

\bibitem{zihayat2017mining}
M.~Zihayat, H.~Davoudi, and A.~An, ``Mining significant high utility gene
  regulation sequential patterns,'' \emph{BMC Systems Biology}, vol.~11, no.~6,
  pp. 109--109, 2017.

\bibitem{fernando2012effective}
B.~Fernando, E.~Fromont, and T.~Tuytelaars, ``Effective use of frequent itemset
  mining for image classification,'' in \emph{Proceedings of European
  Conference on Computer Vision}.\hskip 1em plus 0.5em minus 0.4em\relax
  Springer, 2012, pp. 214--227.

\bibitem{brauckhoff2009anomaly}
D.~Brauckhoff, X.~Dimitropoulos, A.~Wagner, and K.~Salamatian, ``Anomaly
  extraction in backbone networks using association rules,'' in
  \emph{Proceedings of the 9th ACM SIGCOMM International Conference on Internet
  Measurement}, 2009, pp. 28--34.

\bibitem{han2004mining}
J.~Han, J.~Pei, Y.~Yin, and R.~Mao, ``Mining frequent patterns without
  candidate generation: A frequent-pattern tree approach,'' \emph{Data Mining
  and Knowledge Discovery}, vol.~8, no.~1, pp. 53--87, 2004.

\bibitem{zaki2000scalable}
M.~J. Zaki, ``Scalable algorithms for association mining,'' \emph{IEEE
  Transactions on Knowledge and Data Engineering}, vol.~12, no.~3, pp.
  372--390, 2000.

\bibitem{leung2007cantree}
C.~K.~S. Leung, Q.~I. Khan, Z.~Li, and T.~Hoque, ``{CanTree}: a canonical-order
  tree for incremental frequent-pattern mining,'' \emph{Knowledge and
  Information Systems}, vol.~11, no.~3, pp. 287--311, 2007.

\bibitem{chang2003finding}
J.~H. Chang and W.~S. Lee, ``Finding recent frequent itemsets adaptively over
  online data streams,'' in \emph{Proceedings of the 9th ACM SIGKDD
  International Conference on Knowledge Discovery and Data Mining}, 2003, pp.
  487--492.

\bibitem{hong2004fuzzy}
T.~P. Hong, C.~S. Kuo, and S.~L. Wang, ``A fuzzy aprioritid mining algorithm
  with reduced computational time,'' \emph{Applied Soft Computing}, vol.~5,
  no.~1, pp. 1--10, 2004.

\bibitem{srikant1996mining}
R.~Srikant and R.~Agrawal, ``Mining sequential patterns: Generalizations and
  performance improvements,'' in \emph{Proceedings of the 5th International
  Conference on Extending Database Technology}.\hskip 1em plus 0.5em minus
  0.4em\relax Springer, 1996, pp. 1--17.

\bibitem{pei2004mining}
J.~Pei, J.~Han, B.~Mortazavi~Asl, J.~Wang, H.~Pinto, Q.~Chen, U.~Dayal, and
  M.-C. Hsu, ``Mining sequential patterns by pattern-growth: The prefixspan
  approach,'' \emph{IEEE Transactions on Knowledge and Data Engineering},
  vol.~16, no.~11, pp. 1424--1440, 2004.

\bibitem{fournier2017surveyi}
P.~Fournier-Viger, J.~C.~W. Lin, B.~Vo, T.~T. Chi, J.~Zhang, and H.~B. Le, ``A
  survey of itemset mining,'' \emph{Wiley Interdisciplinary Reviews: Data
  Mining and Knowledge Discovery}, vol.~7, no.~4, p. e1207, 2017.

\bibitem{fournier2017surveys}
P.~Fournier-Viger, J.~C.~W. Lin, R.~U. Kiran, Y.~S. Koh, and R.~Thomas, ``A
  survey of sequential pattern mining,'' \emph{Data Science and Pattern
  Recognition}, vol.~1, no.~1, pp. 54--77, 2017.

\bibitem{chan2003mining}
R.~Chan, Q.~Yang, and Y.~D. Shen, ``Mining high utility itemsets,'' in
  \emph{Proceedings of the 8th IEEE International Conference on Data
  Mining}.\hskip 1em plus 0.5em minus 0.4em\relax IEEE Computer Society, 2003,
  pp. 19--19.

\bibitem{qu2019efficient}
J.~F. Qu, M.~Liu, and P.~Fournier-Viger, ``Efficient algorithms for high
  utility itemset mining without candidate generation,'' in \emph{High-Utility
  Pattern Mining}.\hskip 1em plus 0.5em minus 0.4em\relax Springer, 2019, pp.
  131--160.

\bibitem{tseng2012efficient}
V.~S. Tseng, B.~E. Shie, C.~W. Wu, and P.~S. Yu, ``Efficient algorithms for
  mining high utility itemsets from transactional databases,'' \emph{IEEE
  Transactions on Knowledge and Data Engineering}, vol.~25, no.~8, pp.
  1772--1786, 2012.

\bibitem{wang2016efficiently}
J.~Z. Wang, J.~L. Huang, and Y.~C. Chen, ``On efficiently mining high utility
  sequential patterns,'' \emph{Knowledge and Information Systems}, vol.~49,
  no.~2, pp. 597--627, 2016.

\bibitem{gan2020proum}
W.~Gan, J.~C.~W. Lin, J.~Zhang, H.~C. Chao, H.~Fujita, and Y.~P. S, ``{ProUM}:
  Projection-based utility mining on sequence data,'' \emph{Information
  Sciences}, vol. 513, pp. 222--240, 2020.

\bibitem{gan2020fast}
W.~Gan, J.~C.~W. Lin, J.~Zhang, P.~Fournier-Viger, H.~C. Chao, and P.~S. Yu,
  ``Fast utility mining on sequence data,'' \emph{IEEE Transactions on
  Cybernetics}, vol.~51, no.~2, pp. 487--500, 2021.

\bibitem{gan2021utility}
W.~Gan, J.~C.~W. Lin, J.~Zhang, H.~Yin, P.~Fournier-Viger, H.~C. Chao, and
  P.~S. Yu, ``Utility mining across multi-dimensional sequences,'' \emph{ACM
  Transactions on Knowledge Discovery from Data, arXiv:1902.09582}, 2021.

\bibitem{gan2019utility}
W.~Gan, J.~C.~W. Lin, H.~C. Chao, and P.~S. Yu, ``Utility-driven mining of high
  utility episodes,'' in \emph{IEEE International Conference on Big
  Data}.\hskip 1em plus 0.5em minus 0.4em\relax IEEE, 2019, pp. 2644--2653.

\bibitem{wang2018incremental}
J.~Z. Wang and J.~L. Huang, ``On incremental high utility sequential pattern
  mining,'' \emph{ACM Transactions on Intelligent Systems and Technology},
  vol.~9, no.~5, pp. 1--26, 2018.

\bibitem{gan2018privacy}
W.~Gan, J.~C.-W. Lin, H.~C. Chao, S.~L. Wang, and P.~S. Yu, ``Privacy
  preserving utility mining: a survey,'' in \emph{IEEE International Conference
  on Big Data}.\hskip 1em plus 0.5em minus 0.4em\relax IEEE, 2018, pp.
  2617--2626.

\bibitem{han2000mining}
J.~Han, J.~Pei, and Y.~Yin, ``Mining frequent patterns without candidate
  generation,'' \emph{ACM SIGMOD Record}, vol.~29, no.~2, pp. 1--12, 2000.

\bibitem{zhang2020tkus}
C.~Zhang, Z.~Du, W.~Gan, and P.~S. Yu, ``{TKUS}: Mining top-$k$ high-utility
  sequential patterns,'' \emph{arXiv preprint, arXiv:2011.13454}, 2020.

\bibitem{yin2012uspan}
J.~Yin, Z.~Zheng, and L.~Cao, ``{USpan}: an efficient algorithm for mining high
  utility sequential patterns,'' in \emph{Proceedings of the 18th ACM SIGKDD
  International Conference on Knowledge Discovery and Data Mining}, 2012, pp.
  660--668.

\end{thebibliography}

\end{document}